\documentclass[12pt]{article}

\usepackage{lscape}
\usepackage[titletoc,title]{appendix}

\usepackage{epsfig,cancel,amsthm,amssymb,ulem}
\usepackage{color,soul}

\usepackage{wrapfig}
\usepackage{color}
\usepackage{graphicx,framed}
\usepackage[colorlinks=true, pdfstartview=FitV, linkcolor=black, citecolor=black, urlcolor=blue]{hyperref}
\definecolor{shadecolor}{rgb}{0.95, 0.95, 0.86}
\definecolor{darkgreen}{rgb}{0.2, 0.5,  0}

\def\&{\vspace{-5pt}&}

\makeatletter
\@addtoreset{equation}{section}
\makeatother

\def\Tr{ {\rm Tr}}
\def\tr{{\rm tr}}

\def\reg{{\rm reg}}

\def \SS{\mathfrak{S}}
\def \GG{\mathcal{G}}
\def \eqref#1{(\ref{#1})}
\def \& {&\hspace{-10pt}}

\def \wt{\widetilde}

\newcommand{\bt}{\beta}

\newcommand{\G}{\Gamma}

\renewcommand{\d}{\mathrm d}

\newcommand{\p}{\partial}

\newcommand{\br}{{\mathbb R}}

\newcommand{\nn}{\nonumber} 

\newcommand{\g}{\mathfrak{g}}

\newtheorem{theorem}{Theorem}[section]
\newtheorem{example}[theorem]{Example}
\newtheorem{exercise}[theorem]{Exercise}
\newtheorem{conjecture}[theorem]{Conjecture}

\newtheorem{lemma}[theorem]{Lemma}
\newtheorem{remark}[theorem]{Remark}

\newtheorem{proposition}[theorem]{Proposition} 
\newtheorem{corollary}[theorem]{Corollary} 
\newtheorem{definition}[theorem]{Definition}

\def\le{\left}
\def\ri{\right}

\def\bt{\begin{theorem}}
\def\et{\end{theorem}}
\def\bc{\begin{corollary}}
\def\ec{\end{corollary}}
\def\bx{\begin{example}\small}
\def\ex{\end{example}}
\def\bxr{\begin{exercise}\small}
\def\exr{\end{exercise}}
\def\bl{\begin{lemma}}
\def\el{\end{lemma}}
\def\bd{\begin{definition}}
\def\ed{\end{definition}}
\def\bp{\begin{proposition}}
\def\ep{\end{proposition}}

\def\br{\begin{remark}}
\def\er{\end{remark}}

\def\be{\begin{eqnarray}}
\def\ee{\end{eqnarray}}
\def\&{\hspace{-15pt}&}
\def\bea{\begin{eqnarray}}
\def\eea{\end{eqnarray}}
\def\beas{\begin{eqnarray*}}
\def\eeas{\end{eqnarray*}}

\newcommand{\CC}{\mathbb{C}}
\def\L{\mathcal L}

\def\l{\lambda }

\def\1{{\bf 1}}

\newcommand{\h}{\mathfrak{h}}

\def\QED {\hfill $\Box$\par\vskip 5pt}

\newcommand{\Ker}{\mathrm{Ker}}

\newcommand{\ad}{\mathrm {ad}}

\begin{document}
\title{Simple Lie algebras and topological ODEs}
\author{Marco Bertola, Boris Dubrovin, Di Yang}
\date{}
\maketitle
\begin{abstract}
For a simple Lie algebra $\g$ we define  a system of linear ODEs with polynomial coefficients, which we call the topological equation of $\g$-type. 
The dimension of the space of solutions regular at infinity is equal to the rank of 
the Lie algebra. 
For the simplest example $\g=sl_2(\mathbb C)$ the regular solution can be expressed via products of Airy functions and their derivatives; 
this matrix valued function was used in our previous work \cite{BDY1} for computing logarithmic derivatives of the Witten--Kontsevich tau-function.
For an arbitrary simple Lie algebra we construct a basis in the space of regular solutions to the topological equation called {\it generalized Airy resolvents}.
We also outline applications of the generalized Airy resolvents to computing the Witten and Fan--Jarvis--Ruan invariants of the Deligne--Mumford moduli spaces of stable algebraic curves.
\end{abstract}

{\small \noindent \textbf{Keywords.} simple Lie algebra; topological equation; dual topological equation;  generalized Airy resolvents; normal form; essential series.}

\section{Introduction}
In this paper to  an arbitrary simple Lie algebra we will associate two families of special functions defined by certain systems of linear ODEs with polynomial coefficients along with appropriate asymptotic conditions. Let us begin with the first family.

\subsection{Topological ODE and generalized Airy resolvents}

Let $(\g,[\cdot,\cdot])$ be a simple Lie algebra over $\CC$ of rank $n$. We sometimes refer to elements of $\g$ as  matrices, even if  the constructions will not depend on the choice of a matrix realization of the Lie algebra.
Choose $\h$ a Cartan subalgebra of $\g$, $\Pi=\{\alpha_1,\dots,\alpha_n\}\subset\h^*$ a set of simple roots, and $\triangle\subset \h^*$ the root system.
Then $\g$ admits the root space decomposition
\be
\g=\h\oplus \bigoplus_{\alpha\in\triangle} \g_{\alpha}.
\ee
Let $\theta$ be the highest root w.r.t. $\Pi$.
Denote by $(\cdot|\cdot): \g \times\g\rightarrow \mathbb{C}$ the  {\it normalized} Cartan--Killing form \cite{Kac} such that
$(\theta|\theta)=2.$
For any root $\alpha$, let $H_{\alpha}\in \h$ denote the root vector of $\alpha$ w.r.t. $(\cdot|\cdot)$.
Choose a set of Weyl generators $E_i\in\g_{\alpha_i},\,F_i\in\g_{-\alpha_i}$, $H_i={2H_{\alpha_i}}/{(\alpha_i|\alpha_i)}$ satisfying
\bea
[E_i,F_i]=H_i, \quad [H_i,E_j]=A_{ij}\,E_j,\quad [H_i,F_j]=-  \, A_{ij}\,F_j
\eea
where $(A_{ij})$ denotes the Cartan matrix associated to $\g$. Here and below, free Latin indices take integer values from $1$ to $n$ unless otherwise indicated.
Choose $E_{-\theta}\in \g_{-\theta},\, E_\theta \in\g_\theta$. They can be normalized by the conditions  
$(E_\theta\,|\,E_{-\theta})=1$ and $\omega(E_{-\theta})=-E_\theta$, where $\omega:\g\to\g$ is the Chevalley involution. Let $I_+:=\sum_{i=1}^n E_i$ be a principal nilpotent element of $\g$. Define
\be \label{cyc-princ}
\Lambda=I_+ +\lambda \, E_{-\theta}.
\ee
Consider the following differential equation for a $\g$-valued function $M=M(\lambda)$ of an independent variable $\lambda$
\be
M' =  [M,\Lambda], \qquad   '=\frac{d}{d\lambda}. \label{int1}
\ee

\begin{definition} We call eq. \eqref{int1} the topological differential equation 
(in short: {\it topological ODE})
  {\it of $\g$-type}.
\end{definition}

Solutions of the differential equation \eqref{int1} are entire functions of the complex variable $\lambda$. At $\lambda=\infty$ they may have a singularity.  A solution $M(\lambda)$ is called {\it regular} 
if it grows at most polynomially at  $|\lambda|\to\infty$ within a certain chosen sector of the complex $\lambda$-plane. 
Denote by $\SS_\infty^{\reg}(\g)$ the vector space of regular solutions to \eqref{int1}. Our goal is to describe the space $\SS_\infty^\reg(\g)$ for an arbitrary simple Lie algebra $\g$. Actually, we will only deal with asymptotic expansions of regular solutions without entering into details about the region of their validity that can be described by means of standard techniques of asymptotic analysis of ODEs, see e.g. \cite{Wasow}.

As a way of introductory example we consider now the simplest instance  $\g=sl_2(\mathbb C)$. In this case $\Lambda=\left(\begin{array}{cc} 0 & 1 \\ \lambda & 0\end{array}\right)$. Write $M(\lambda)=\left(\begin{array}{cr} a(\lambda) & b(\lambda)\\ c(\lambda) & -a(\lambda)\end{array}\right)$. The equation \eqref{int1} reads
\be\label{top1}
\left(\begin{array}{c}a'\\ b' \\ c'\end{array}\right) =\left(\begin{array}{rrr}0 & \lambda & -1\\ 2 & 0 & 0\\ -2\lambda & 0 & 0\end{array}\right)\left(\begin{array}{c}a\\ b \\ c\end{array}\right).
\ee
The type of behaviour at infinity is essentially specified by the eigenvalues of the matrix of coefficients. In this case they are equal to 0 and $\pm 2\sqrt{\lambda}$. Solutions corresponding to the nonzero eigenvalues have an 
essential singularity $\sim \lambda^{-\frac12}e^{\pm \frac43 \lambda^{3/2}}$ at infinity. The regular solution is unique up to a constant factor. It corresponds to the eigenvector with eigenvalue zero. The solution has the form
\be\label{faza}
M(\lambda)=\frac{\lambda^{-\frac12}}{2}\left(
\begin{array}{cc}
\sum_{g=1}^\infty \frac{(6g-5)!!}{96^{g-1}\cdot (g-1)!} \lambda^{-3g+2} & 2 \sum_{g=0}^\infty \frac{(6g-1)!!}{96^g\cdot g!} \lambda^{-3g}\\
\\
-2 \sum_{g=0}^\infty\frac{6g+1}{6g-1} \frac{(6g-1)!!}{96^g\cdot g!} \lambda^{-3g+1} &  -\sum_{g=1}^\infty \frac{(6g-5)!!}{96^{g-1}\cdot (g-1)!} 
\lambda^{-3g+2}\\
\end{array}
\right).
\ee
This matrix-valued function appeared in our paper \cite{BDY1}. It was used, in a slightly modified normalization, as key tool within an efficient algorithm for computation of intersection numbers of tautological classes on the Deligne--Mumford moduli spaces, see eq. \eqref{KdV-N-point} below.

The topological ODE  for the Lie algebra $sl_2(\mathbb C)$ is closely related to the theory of Airy functions. Indeed, from \eqref{top1} it readily follows that
\be\label{airy4}
a=\frac12 b', \quad c=-\frac12 b''+\lambda\, b
\ee
while for $b=b(\lambda)$ one arrives at a third order ODE
\be\label{airy3}
-\frac12 b'''+2\lambda\, b'+b=0.
\ee
Solutions to this equation are products of solutions to the Airy equation
\be\label{airy2}
y''=\lambda\, y.
\ee
That is, $b(\lambda)$ is the diagonal value of the resolvent of the Airy equation.
The matrix $M(\lambda)$ can also be considered as resolvent of the matrix version of the Airy equation
\be \label{Airy-vector}
\left[\frac{d}{d\lambda}+\left( \begin{array}{cc} 0 & 1\\ \lambda & 0\end{array}\right)\right]\vec{y}=0, \quad 
\vec{y}=\left(\begin{array}{r}y \\ -y'\end{array}\right).
\ee

Let us proceed now with the general case. The eigenvectors of the matrix of coefficients of the system \eqref{int1} with zero eigenvalue correspond to matrices commuting with $\Lambda$. As observed in \cite{Kostant}, the kernel of the linear operator
$$
{\rm ad}_{\Lambda(\lambda_0)}: \g \to \g
$$
for any $\lambda_0\in \mathbb{C}$ has dimension $n={\rm rk}\,\g$. This suggests that the dimension of the space of regular solutions to the topological ODE \eqref{int1} is equal to $n$. We show in Thm. \ref{Qinf} that this is the case indeed. Moreover, we will construct a particular basis in the space $\SS_\infty^\reg(\g)$. To this end denote $L(\g)=\g \otimes \CC[\lambda, \lambda^{-1}]$ the loop algebra of $\g$. Introduce the {\it principal gradation} on the Lie algebra $L(\g)$ in the following way
\be\label{prince}
\deg E_i=1, ~\deg H_i=0, ~ \deg F_i=-1,~i=1, \dots, n, \quad \deg\lambda=h.
\ee
Here $h$ is the Coxeter number of $\g$. According to this gradation we have
$
\deg \Lambda=1.
$
Recall \cite{Kac1978} that $\Ker \, \ad_\Lambda\subset L(\g)$ has the following decomposition 
\be
\Ker \, \ad_\Lambda=\bigoplus_{j \in E} \mathbb{C} \Lambda_j,\qquad \deg\Lambda_j=j\in E:=\bigsqcup_{i=1}^n (m_i+h \mathbb{Z})
\ee
where the integers
\be\label{expo}
1=m_1<m_2\leq \dots \leq m_{n-1}<m_n=h-1
\ee
are the exponents of $\g$. The matrices $\Lambda_i$ commute pairwise
\be\label{commute}
[\Lambda_i,  \Lambda_j ]=0,\qquad \forall\,i,j\in E.
\ee
They can be normalized in such a way \cite{KW,Kac} that
\bea
\&\& \Lambda_{m_a+kh}=\Lambda_{m_a}\, \lambda^k,\qquad k\in \mathbb{Z},\label{norm-Lambda-1}\\
\&\& (\Lambda_{m_a}|\,\Lambda_{m_b})=h \,\lambda\,\delta_{a+b, n+1}. \label{norm-Lambda-2}
\eea
In particular, we choose $\Lambda_1=\Lambda.$

\begin{theorem} \label{Qinf}
1) $\dim_\CC\SS_\infty^\reg(\g)={\rm rk}\, \g.$

2) There exists a unique basis $M_1(\lambda)$, \dots, $M_n(\lambda)$ in the space $\SS_\infty^\reg(\g)$ of the form
\begin{eqnarray}\label{baza1}
&&
M_a(\lambda) = \lambda^{-\frac{m_a}{h}} \left[ \Lambda_{m_a} +\sum_{k=1}^\infty M_{a,k}(\lambda)\right],\nn\\
&&
\\
&&
M_{a,k}(\lambda)\in L(\g),\quad \deg M_{a,k}(\lambda)=m_a-(h+1)k, \quad a=1, \dots, n.
\nn
\eea
\end{theorem}

\noindent The proof is given in Sect.\,\ref{proof}. 

\begin{definition} The solutions \eqref{baza1} to the topological ODE of $\g$-type are called {\it generalized Airy resolvents} of $\g$-type.
\end{definition}

The topological ODE \eqref{int1} can be recast into the form
\be\label{int11}
\left[ \frac{d}{d\lambda}+\Lambda, M\right]=0.
\ee
Using such a representation along with \eqref{commute}, \eqref{norm-Lambda-2} as well as the invariance property of the Cartan--Killing form one arrives at

\begin{proposition} \label{comm-p}
The generalized Airy resolvents commute pairwise
\be\label{commute1}
\left[ M_a(\lambda), M_b(\lambda)\right]=0, \quad a, \, b=1, \dots, n.
\ee
Moreover,  $(M_a(\lambda)\,|\,M_b(\lambda))=h\,\delta_{a+b, n+1}.$
\end{proposition}

\begin{example} 
\label{An-recursion} 
Consider the $A_n$ case, $\g=sl_{n+1}(\mathbb C)$. The $(n+1)\times (n+1)$ matrix $\Lambda$ reads
$$
\Lambda=\left(\begin{array}{ccccc} 
0 & 1 & 0 & \cdots & 0\\
0 & 0 & 1 & \ddots & 0\\
\vdots & \ddots & \ddots & \ddots & \vdots\\
0 & 0 & \ddots & \ddots & 1\\
\lambda & 0 & 0 & \cdots & 0
\end{array}\right).
$$
The exponents are $m_a=a$, $a=1, \dots, n$, the Coxeter number is $h=n+1$. One can choose
$$
\Lambda_i=\Lambda^i, \quad i=1, \dots, n.
$$
Then
\be \label{first-order}
M_a(\lambda) =\lambda^{-\frac{a}{n+1}}\left[ \Lambda^a+\frac1{(n+1)\, \lambda} \sum_{\nu=1}^n  \frac{\zeta^{\nu\,a}-1}{4\sin^2 \frac{\pi \nu}{n+1}}\Omega_\nu\Lambda^{a-1}+\dots\right]
\ee
where $\zeta=e^{\frac{2\pi i}{n+1}}$ is a root of unity,
$
\Omega_\nu={\rm diag}\left(1, \zeta^\nu, \zeta^{2\nu}, \dots , \zeta^{n\, \nu}\right).
$
The recursion procedure for computing subsequent terms will be explained below, see Example \ref{primeran}. 
In the particular case $n=1$ one obtains the solution \eqref{faza}.
\end{example}

The recursive procedure for computing the series expansions of the generalized Airy resolvents outlined in the Example \ref{An-recursion} above can be represented in an alternative form. The basic idea can be already seen from the $sl_2(\mathbb C)$ example: the topological ODE for a matrix-valued function reduces to a scalar differential equation \eqref{airy3} for one of the entries of the matrix (called $b(\lambda)$ above); other matrix entries are expressed in terms of $b(\lambda)$ and its derivatives, see eq. \eqref{airy4}. 

Let us explain a general method  for reducing the topological ODE for an arbitrary simple Lie algebra of rank $n$ to a system of ODEs with $n$ dependent variables. To this end we will use the Jordan decomposition of the operator $\ad_{I_+}$ described hereafter.

Let $\rho^\vee\in \h$ be the Weyl co-vector of $\g$, whose defining equations are
\be
\alpha_i(\rho^\vee)=1,\qquad i=1,\dots,n.
\ee
Write $\rho^\vee=\sum_{i=1}^n x_i \,H_i,\,x_i\in\CC$ and define 
$
I_-=2\sum_{i=1}^n x_i \,F_i.
$
Then $I_+,I_-,\rho^\vee$ form an $sl_2(\mathbb{C})$ Lie algebra:
\be
\label{SL2princ}
[\rho^\vee, I_+]= I_+,\quad [\rho^\vee,I_-]=-I_-,\quad [I_+,I_-]=2\rho^\vee.
\ee
According to \cite{Kostant,BFRFW},  there exist elements $\gamma^1,\dots,\gamma^n\in \g$ such that
\be
{\rm Ker} \, {\rm ad}_{I_-}={\rm Span}_{\mathbb{C}} \{\gamma^1,\dots,\gamma^n\},\qquad  [\rho^\vee,\gamma^i]=-m_i \, \gamma^i.
\ee
Fix $\{\gamma^1,\dots,\gamma^n\}$, 
then the lowest weight decomposition of $\g$ has the form
\bea
\g=\bigoplus_{i=1}^n \L^i,\qquad \L^i={\rm Span}_{\mathbb{C}} \{\gamma^i,\ad_{I_+} \gamma^i, \dots, \ad_{I_+}^{2m_i} \gamma^i\}.
\eea
Here each $\L^i$ is an $sl_2(\mathbb{C})$-module. Any $\g$-valued function
$M(\lambda)$ can be uniquely represented in the form 
\be\label{;;}
M(\lambda)=\sum_{i=1}^n S_{i}(\lambda)  \, \ad_{I_+}^{2m_i} \gamma^i+\sum_{i=1}^n \sum_{m=0}^{2m_{i}-1} K_{im}(\lambda) \, \ad_{I_+}^m \gamma^i
\ee
where $S_{i}(\lambda)$, $K_{im}(\lambda)$ are certain complex-valued functions.
\bt \label{ce-thm}
Let  $M$ be any solution of \eqref{int1}. The functions $K_{im}$ have the following expressions
\bea
\&\& \!\!\! K_{im} = (-1)^m S_i^{(2m_i-m)}+\sum_{u=1}^n \sum_{ v=0}^{ 2m_i-1-m}  k_{imuv}^0  \, S_u^{(v)}+ \lambda\, \sum_{u=1}^n \sum_{ v =0}^{ m_i-1-m}  k_{imuv}^1 \, S_u^{(v)},~ m=0,\dots,m_i-1;\nn\\
\&\& \!\!\! K_{im} =(-1)^m \, S_i^{(2m_i-m)},~ m=m_i,\dots, 2m_i-1\nn\\
\label{Kim-S}
\ee
with constant coefficients $k_{imuv}^0,\,k_{imuv}^1$ independent of the choice of the solution $M.$ Moreover, the topological ODE \eqref{int1} is equivalent to a system of linear ODEs for
$S_1,\dots,S_n$ of the form
\be
\label{S-red}
S_i^{(2m_i+1)}=\sum_{u=1}^n \sum_{ v=0}^{ 2m_i}  k_{iuv}^0  \, S_u^{(v)}+ \lambda\, \sum_{u=1}^n \sum_{ v =1 }^{ m_i}  k_{iuv}^1 \, S_u^{(v)}, \quad i=1, \dots, n.
\ee
Here $k_{iuv}^0,\,k_{iuv}^1$ are constants independent of the choice of the solution $M.$
\et
\noindent The proof is given in Sect.\,\ref{c-e-s}. Note that the coefficients
$k^0_{iuv}$, $k^1_{iuv}$ coincide with $k^0_{imuv}$, $k^1_{imuv}$ at $m=-1$. See details in the proof.

We call eqs. \eqref{S-red} the {\it reduced topological ODEs} of $\g$-type.
Denote by
$
S_{a;1}(\lambda), \dots, S_{a;n}(\lambda)
$ 
the $S$-coefficients in the decomposition \eqref{;;} of the generalized Airy resolvent $M_a(\lambda)$, $a=1, \dots, n$. 
\begin{definition} 
The series $S_{a; i}(\lambda)$ are called {\it essential series} of $\g$-type.
\end{definition}

\begin{example} [The $A_1$ case] 
The essential series $S_{1;1}$ coincides with the $(1,2)$-entry of the Airy resolvent \eqref{faza} (the function $b(\lambda)$ in \eqref{top1}). 
The expression \eqref{Kim-S} is given explicitly by \eqref{airy4}.

\end{example}
\begin{example} [The $A_2$ case]  \label{A2-intro}
The decomposition \eqref{;;} reads
\be
M=\left(\begin{array}{ccc} 2 K_{11} +K_{22} & - 2 S_1 -3 K_{23} & 6 S_2\\
2 K_{10}+K_{21} & -2 K_{22} & 3 K_{23}-2 S_1\\
K_{20} & 2 K_{10} -K_{21} & -2 K_{11} +K_{22}\end{array}\right).
\ee
The expression \eqref{Kim-S} is given explicitly by
\bea
\&\&K_{10}= S_1''+ 3\lambda \, S_2, \quad K_{11}= -S_1',\\
\&\&K_{20}= S_2^{(4)}-2\lambda\,   S_1,~~ K_{21}= -S_2''', ~~
K_{22}= S_2'',~~ K_{23}= -S_2'.
\eea
The reduced equations \eqref{S-red} read
\bea
\&\&  6 S_2+2 S_1^{(3)}+9 \, \lambda \, S_2'=0,\\
\&\& 2 S_1+2 S_1^{(4)}+15 S_2'
-S_2^{(5)}+\lambda \,  \left(9 S_2''+6 S_1'\right)=0.
\eea
Explicit expressions of the essential series $S_{a;j}$ of the $A_2$-type can be found in Sect.\,\ref{rem}.
\end{example}

\begin{proposition}  \label{S-prop} 
1. For any $a\in\{1,2,\dots n\}$, the essential series $S_{a;i}$ satisfy 
\bea
\&\& 1) ~~ S_{a;1}(\lambda), \dots, S_{a;n}(\lambda)\in \lambda^{-\frac{m_a}{h}}  \cdot \mathbb C[[\lambda^{-1}]];\nn \\
\&\& 2) ~~ \mbox{the vectors } (S_{a;i})_{i=1\dots n},\,a=1,\dots,n \mbox{ are linearly independent};\nn\\
\&\& 3) ~~ S_{a;a} \mbox{ is non-zero}.\nn
\eea

2. Let $\g\neq D_{n=2k}$ and $M$ be a regular solution of the topological ODE \eqref{int1} such that the $S_a$-coefficient of $M$ coincides with $S_{a;a}$; then $M=M_a.$
If $\g$ is of type $D_{n}$ with $n$ even, 
the statement remains valid for $a\not\in\{n/2, n/2+1\}$, and it should be modified for 
$a=n/2$ and $a=n/2+1$ in the following way: $M=M_a$ if the
 $S_{n/2},\,S_{n/2+1}$-coefficients of $M$ coincide with $S_{a;n/2}$, $S_{a;n/2+1}$, respectively.
\end{proposition}
\noindent The proof will be given in Sect.\,\ref{c-e-s}. We call $S_{a;a},\,a=1,\dots,n$ the {\it fundamental series} of $\g$-type.

\subsection{Dual topological ODE and its normal form}

Solutions to \eqref{int1} can be expressed by a suitable {\it integral transform} via another class of functions satisfying somewhat simpler differential equations. These functions will be constructed as solutions to another system of linear ODEs for a $\g$-valued function $G=G(x)$, $x\in\mathbb C$. In the notations of the previous section it reads
\be \label{dual-LM}
[G',E_{-\theta}]+[G,I_+]+x\,G=0, \qquad {}'=\frac{d}{dx}.
\ee

\begin{definition}
We call equation \eqref{dual-LM} the {\it dual topological differential equation} (dual topological ODE for short). 
The space of solutions of \eqref{dual-LM} is denoted by $\mathcal{G}(\g).$
\end{definition}

Solutions to the dual topological ODE are related with the generalized Airy resolvents by the following Laplace--Borel transform 
\be\label{lb-trans}
M(\lambda)=\int_C \, G (x) \, e^{-\lambda\,x} dx,
\ee
where $C$ is a suitable contour on the complex $x$-plane, which can depend on the choice of a solution $M(\lambda)\in\SS_\infty^\reg(\g)$.

\begin{example} 
For $\g=sl_2(\mathbb C)$ the dual topological equation for the matrix-valued function $G=\left(\begin{array}{cr} \alpha(x) & \beta(x)\\ \gamma(x) & -\alpha(x)\end{array}\right)$ reduces to
\be \label{A1-normal-form}
\beta'=\left( \frac{x^2}4 -\frac1{2x}\right) \beta.
\ee
This yields the following integral representation\footnote{
{Here and below the Laplace-type integrals are understood in a formal sense by applying a term-by-term integration
$$
\frac1{\sqrt{\pi}}\int\limits_0^\infty e^{\frac{x^3}{12}-\lambda\, x} \frac{dx}{\sqrt{x}}=\frac1{\sqrt{\pi}}\sum_{k=0}^\infty \int\limits_0^\infty \frac1{k!}\left(\frac{x^3}{12}\right)^k e^{-\lambda\, x} \frac{dx}{\sqrt{x}}=\frac1{\sqrt{\pi}}\sum_{k=0}^\infty \frac{\Gamma\left(3k+\frac12\right)}{k! \, 12^k} \lambda^{-3k-\frac12}.
$$
In order to arrive at an integral representation of a solution with the same asymptotic expansion within a certain sector near $\lambda=\infty$ one can integrate over a ray of the form $x=r\,e^{\frac{(2k+1)\pi\,i}3}$, $0<r<\infty$ for $k=0, \, 1,$ or $2$.
}} for solutions to the differential equation \eqref{airy3}
\be\label{contour}
b(\lambda)=\frac1{\sqrt{\pi}}\int\limits_0^\infty e^{\frac{x^3}{12}-\lambda\, x} \frac{dx}{\sqrt{x}}.
\ee

\end{example}

For an arbitrary simple Lie algebra we will now describe the reduction procedure for the 
dual topological ODE similar to the one used in the previous subsection. Write
\be\label{dual-M}
G(x)=\sum_{i=1}^n \phi_{i}(x)  \, \ad_{I_+}^{2m_i} \gamma^i+\sum_{i=1}^n \sum_{m=0}^{2m_{i}-1} \wt K_{im}(x) \, \ad_{I_+}^m \gamma^i
\ee
(cf. \eqref{;;} above).
\bt \label{ce-p} 
i). Let $G$ be any solution of \eqref{dual-LM}. Then
$\wt K_{im}$ have the following expressions
\be \label{Kim-dual}
\&\& \wt K_{im}= \sum_{u=1}^n \sum_{ v=0}^{2m_i-m} 
 (-x)^v\,  k_{imuv}^0  \, \phi_u+ \sum_{u=1}^n \sum_{ v=0}^{m_i-1-m}  (-1)^v  \, k_{imuv}^1 \, \le(x^v\, \phi_u'+v\,x^{v-1} \phi_u \ri),~m=0,\dots,m_i-1,\nn\\
 \&\& \wt K_{im} =x^{2m_i-m} \, \phi_i,~m=m_i,\dots,2m_i
\ee
where $k_{imuv}^0,\,k_{imuv}^1$ are the same constants as in Thm.\,\ref{ce-thm}.
Moreover, the dual topological ODE \eqref{dual-LM} is equivalent to  a system of {\bf first order} ODEs 
\be \label{first-dual}
x^{2m_i+1}\, \phi_i+\sum_{u=1}^n \sum_{ v=0}^{ 2m_i}  (-1)^v \, k_{iuv}^0  \, x^v\, \phi_u+  \sum_{u=1}^n \sum_{ v =1 }^{ m_i} (-1)^v\,  k_{iuv}^1 \, (v\,x^{v-1}\,\phi_u+x^v\,\phi_u')=0.
\ee
Here, $i=1, \dots, n,$ and $k^0_{iuv},\,k^1_{iuv}$ are constants (the same as in Thm.\,\ref{ce-thm}).

ii). The point $x=0$ is a Fuchsian singularity of the ODE system
\eqref{first-dual}. More precisely, equations \eqref{first-dual} have the following {\bf normal form}
\bea
&& \phi' = \frac{V_{-1}}{x}\, \phi + \sum_{k=0}^{2h-2}\, x^k \, V_k \, \phi,\qquad \phi=(\phi_1,\dots,\phi_n)^T, \label{normal-form}\\
&& V_{-1}={\rm diag}(-m_{n+1-a}/h)_{a=1,\dots,n},\quad V_k\in {\rm Mat}(n,\CC),\,k\geq 0.
\eea

iii). $\dim_\CC \GG(\g)={\rm rk}\,\g.$
\et

\begin{remark}
Choosing  the contour $C$ in several ways in \eqref{lb-trans} one obtains a complete basis of solutions of \eqref{int1}.
\end{remark}

\begin{example} \label{A1-normal}
For $\g=sl_2(\CC)$
one reads from \eqref{A1-normal-form} that 
$V_{-1}=-\frac12,\,V_0=V_1=0,\,V_2=\frac14.$
\end{example}
\begin{example} \label{A2-normal}
For $\g=sl_3(\CC)$ the normal form \eqref{normal-form} of the dual topological equation reads
\be
\phi'=\left[ -\frac1{3 x}\left( \begin{array}{cc} 2 & 0\cr 0 & 1\end{array}\right)-\frac{x^2}{9}\left(\begin{array}{cc}0 & 0\cr 2 & 0\end{array}\right) +\frac{x^4}6 \left(\begin{array}{cc} 0 & 1\cr 0 & 0\end{array}\right)\right]\phi.
\ee
\end{example}

Let $G_a$ denote the Laplace--Borel transform of $M_a$, and
$\phi_{a;1}, \dots, \phi_{a;n}$ 
the $\phi$-coefficients of $G_a$, $a=1, \dots, n$. 
Then the matrix $\Phi$ defined by $\Phi_{ia}:=\phi_{a;i}$ is a fundamental solution 
matrix of the normal form \eqref{normal-form}.
\begin{definition}
We call $\phi_{a;i}$ the {\it dual essential series}, and $\phi_{a;a}$ the {\it dual fundamental series}.
\end{definition}
\begin{conjecture}\label{conjIII} 
For any $a\in\{1,\dots,n\},$
the dual fundamental series $\phi_{a;a}$ satisfies a scalar linear ODE of order less than or equal to $n$ with polynomial coefficients in $x.$  
Such an ODE with the smallest possible order, say $r_a$, is called a {\it dominant equation} and $r_a$ is called a {\it dominant number}. The point $x=0$ is a regular singularity of the dominant equation. If $h$ is a prime number, then $r_1=\dots=r_n=h-1.$
\end{conjecture}

For low ranks, the above conjecture is verified, and the dual fundamental series $\phi_{a;a}$ can be expressed via elementary functions or combinations of elementary and Bessel 
functions, see Table \ref{series}. 
However, already for the case of the Lie algebra $sl_5(\mathbb C)$ 
we were not able to identify the dual fundamental series with a classical special function.

\begin{table}[h]
\begin{center}
    \begin{tabular}{| c | c |}
    \hline
    Type of $\g$ &  Special functions arising in dual fundamental series \\ \hline
    $A_1$ & exponential function\\ \hline
    rank $2$ & Bessel functions\\ \hline
    rank $3$ & Bessel functions, exponential functions\\ \hline
    $A_4$ & solutions to scalar ODEs of order $4$ \\ \hline
    $D_4$ & Bessel functions, exponential functions\\ \hline
    \end{tabular}
\end{center}
\caption{Special functions arising from the dual topological ODE of $\g$-type.} \label{series}
\end{table}

{\it Last but not least: why do we call eq. \eqref{int1} the topological differential equation?} For the case of $\g=sl_2(\mathbb C)$ the Airy resolvent \eqref{faza} appears in the expression of \cite{BDY1} for generating series of intersection numbers of $\psi$-classes on the Deligne--Mumford moduli spaces $\overline{\mathcal M}_{g,N}$; see also the formula \eqref{KdV-N-point} below. It turns out that for a simply-laced Lie algebra $\g$ a representation similar to that of \cite{BDY1} can be used for computing the E.\,Witten and H.\,Fan--T.\,Jarvis--Y.\,Ruan (FJRW) intersection numbers. We give the precise expressions in the last section of the present paper, see Theorems \ref{r-spin-thm} and \ref{N-DS-p-real} below. Connection between FJRW invariants and tau-functions of the Drinfeld--Sokolov hierarchy of $A\,D\,E$-type conjectured in \cite{W2} and proven in \cite{FSZ} and \cite{FJR} plays an important role in derivation of these expressions. The non-simply-laced analogue of the FJRW intersection numbers, according to the recent paper by S.-Q.\,Liu, Y.\,Ruan and Y.\,Zhang \cite{LRZ}, is also related to tau-functions of the Drinfeld--Sokolov hierarchies of $B\, C\, F\, G$-type.

\paragraph{Organization of the paper.} 
In Sect.\,\ref{proof} we prove Theorem \ref{Qinf}.  
In Sect.\,\ref{c-e-s} we prove Thm.\,\ref{ce-thm}, Prop.\,\ref{S-prop} and Thm.\,\ref{ce-p}.
In Sect.\,\ref{examples}, we calculate the dual essential series, in particular the dual fundamental series for several special cases. Concluding remarks are given in Sect.\,\ref{rem}.
We give  definition and examples of generalized Airy functions in Appendix A.

\paragraph{Acknowledgements.} 
We would like to thank Philip Candelas, Yassir Dinar, Si-Qi Liu, Youjin Zhang, Jian Zhou for helpful discussions. D.\,Y. is grateful to Youjin Zhang for his advising. The work is partially supported by PRIN 2010-11 Grant ``Geometric and analytic theory of Hamiltonian systems in finite and infinite dimensions" of Italian Ministry of Universities and Researches. M.\,B. is supported in part by the Natural Sciences and Engineering Research Council of Canada (NSERC) and  the Fonds de recherche Nature et t\'echnologie du Qu\'ebec (FQRNT).

\section{Proof of Theorem \ref{Qinf}} \label{proof}

Before proceeding with the proof, we introduce the {\it gradation operator}
\be\label{grad}
{\rm gr}=h\, \lambda \frac{d}{d\lambda}+\ad_{\rho^\vee}.
\ee
Recall that $\rho^\vee\in \h$ is the Weyl co-vector. The following statement \cite{KW} will be useful in the proof.

\begin{lemma} An element $a\in L(\g)$ is homogeneous of principal degree $k$ {\it iff}
\be\label{grad1}
{\rm gr}\, a=k\, a.
\ee
\end{lemma}

The action of the gradation operator can be obviously extended to expressions containing fractional powers of $\lambda$.

The following lemma 
will also be useful.
\bl \label{orto}The following formul\ae\ hold true:
\be
L(\g)=\Ker \, \ad_\Lambda \oplus {\rm Im} \, \ad_\Lambda,\qquad  {\rm Im}\, \ad_\Lambda =(\Ker \, \ad_\Lambda)^{\perp},
\ee
where the orthogonality is with respect to the Cartan--Killing form.
\el
\noindent The proof can be found e.g. in \cite{DS}.


\noindent {\it Proof of Thm. \ref{Qinf}}. Let us fix an exponent $m_a$ for some $a\in \{1,\dots,n\}.$
We first prove existence and uniqueness of formal solutions of \eqref{int1} of the following form
\bea
&& M(\lambda)= \lambda^{-\frac{m_a}h}\sum_{k=0}^\infty M_k(\lambda),\quad  M_0(\lambda)=\Lambda_{m_a} \label{e-37}\\
&& M_k(\lambda)\in L(\g), \quad \deg M_k(\lambda)=m_a-(h+1)k \label{e-38}
\eea
Substituting \eqref{e-37} in \eqref{int1} and equating the terms of equal principal degrees we obtain
\bea
&& [\Lambda,M_0]=0,\\
&& [\Lambda,M_{k+1}]=- M_k'+\frac{m_a}{h} \frac1{\lambda} M_k,\qquad k\geq 0. \label{m-r}
\eea
The first condition follows due to the choice $M_0=\Lambda_{m_a}$  (see eq. \eqref{commute} above). In order to resolve eq. \eqref{m-r} for $k=0$ one has to check that
$$
h\, \lambda\, \Lambda_{m_a}' -{m_a}  \Lambda_{m_a}\in {\rm Im}\, \ad_\Lambda.
$$
Indeed, since $\deg \Lambda_{m_a}=a$ we have
$$
{\rm gr} \,\Lambda_{m_a}\equiv h\, \lambda \, \Lambda_{m_a}'+\left[ \rho^\vee, \Lambda_{m_a}\right]=m_a\Lambda_{m_a}.
$$
So, the equation \eqref{m-r} for $k=0$ becomes
$$
\left[ \Lambda, M_1\right]= \frac1{h\, \lambda}\left[ \rho^\vee, \Lambda_{m_a}\right].
$$
Due to Lemma \ref{orto} it suffices to check that the r.h.s. is orthogonal to ${\rm Ker}\, \ad_\Lambda$. Indeed,
$$
\left( \left[ \rho^\vee, \Lambda_{m_a}\right]\,|\, \Lambda_j\right)=\left(\rho^\vee\, |\,\left[ \Lambda_{m_a}, \Lambda_j\right]\right)=0.
$$
This completes the first step of the recursive procedure.

Assume, by induction, that the $k$-th term $M_k(\lambda)$ exists and, moreover, it satisfies
$$
- M_k'+\frac{m_a}{h} \frac1{\lambda} M_k\in {\rm Im}\, \ad_\Lambda.
$$
Then there exists a matrix $M_{k+1}$ satisfying eq. \eqref{m-r}. It is determined uniquely modulo an element in ${\rm Ker}\, \ad_\Lambda$. Write
\bea
&&M_{k+1}=A+B, \quad A\in {\rm Im}\, \ad_\Lambda,\, B \in {\rm Ker}\,\ad_\Lambda,\\
&& \deg\, A=\deg\, B= m_a-(h+1)(k+1). \label{degab}
\eea
Since the map
$$
\ad_\Lambda: {\rm Im}\, \ad_\Lambda\to {\rm Im}\, \ad_\Lambda
$$
is an isomorphism 
it remains to prove that the matrix $B=B(\lambda)$ is uniquely determined by the condition
$$
- M_{k+1}'+\frac{m_a}{h} \frac1{\lambda} M_{k+1}\in {\rm Im}\, \ad_\Lambda.
$$
Like above, using the degree condition \eqref{degab} along with the gradation operator \eqref{grad} recast the last condition into the form
\be\label{AB}
\left[ \rho^\vee, A+B\right]+(h+1)(k+1)(A+B) \perp {\rm Ker}\, \ad_\Lambda.
\ee
By assumption $A\perp {\rm Ker}\, \ad_\Lambda$. Also $\left[ \rho^\vee, B\right]\perp {\rm Ker}\, \ad_\Lambda$ due to commutativity of ${\rm Ker}\, \ad_\Lambda$. Thus eq. \eqref{AB} reduces to the system
\be\label{AB1}
(h+1)(k+1)\left( B\, |\, \Lambda_j\right) =-\left( \left[ \rho^\vee, A\right]\, |\, \Lambda_j\right), \quad j\in E.
\ee
Decompose now $B$ as follows:
$$
B=\sum_{b=1}^n c_b(\lambda)\, \Lambda_{m_b}(\lambda).
$$
Choosing $j=m_1, \dots, m_n$ and using the normalization \eqref{norm-Lambda-2} one finally obtains
$$
B=-\frac1{\lambda\, h\, (h+1) (k+1)} \sum_{b=1}^n \left( \left[\rho^\vee, A\right]\, |\, \Lambda_{m_{n+1-b}}\right) \Lambda_{m_b}.
$$
This concludes the inductive step and thus completes the proof of existence and uniqueness of solutions \eqref{baza1} to the topological ODE \eqref{int1}.

The series $M_1(\lambda),\dots,M_n(\lambda)$ are linearly independent since their leading terms are.
It is known from \cite{Wasow} that there exist analytic solutions to the topological ODE whose asymptotic expansions in certain sector coincide with $M_i$. 
These solutions we denote again by $M_i$. 
We conclude that 
\be\label{one-side}
{\rm Span}_{\CC} \{M_1,\dots,M_n\}\subset \SS_\infty^\reg(\g).
\ee
In particular $\dim_\CC\SS_\infty^\reg(\g)\geq n.$


On another hand, according to Kostant \cite{Kostant} for any $\lambda_0\in\CC$ the kernel of 
$
{\rm ad}_{\Lambda(\lambda_0)}: \g \to \g
$
 is of dimension $n$. This implies that 
\be
\dim_\CC \SS_\infty^\reg\leq n.
\ee

Hence $\dim_\CC\SS_\infty^\reg(\g)=n$ and $\SS_\infty^\reg(\g)={\rm Span}_\CC \{M_1,\dots,M_n\}.$ Thm.\,\ref{Qinf} is proved. $\hfill{\square}$

The above proof gives an algorithm for computing $M_a,\,a=1,\dots,n$ recursively.

\br 
Thm.\,\ref{Qinf} and Prop.\,\ref{comm-p} tell that $\SS_\infty^\reg(\g)$ forms a one-parameter family of Abelian subalgebras of $\g.$
\er

\begin{example}[The $A_n$ case]   \label{primeran}
Take the same matrix realization of $\g$ 
 as in Example \ref{An-recursion}, 
so that
 the normalized Cartan--Killing form coincides with the matrix trace.
Recall that $h=n+1,\,m_a=a$.  We have
\bea
&&  \Lambda_{j} = \Lambda^j,\qquad \rho^\vee= {\rm diag} \le(\frac {n}2, \frac {n-2}2 ,\dots, -\frac n2 \ri),\\
&& H_j=E_{j,j}-E_{j+1,j+1},\qquad \h={\rm Span}_\CC \{H_j\}_{j=1}^n.
\eea
Fix an exponent $m_a$. Consider $[L,M]=0,\,L=\p_{\lambda}+\Lambda$ and write $M=\lambda^{-\frac{a}h}\sum_{k=0}^\infty M_k,$ then
\be
M_k = A_k + B_k,\quad A_k\in {\rm Im}\, \ad_\Lambda,\, B_k \in {\rm Ker}\,\ad_\Lambda,~k\geq 0.
\ee
Here $A_k$ and $B_k$ have the principal degree $\deg A_k=\deg B_k=a-(h+1)k,$ and
\be 
B_k= \frac {-1}{ \lambda\,k\,h (h+1) }  \sum_{b=1}^{n}  {(\ad_{\rho^\vee} A_{k} \,|\, \Lambda_{n+1-b} )} \, \Lambda_b,\quad k\geq 1.
\ee
It can be easily seen that 
\bea
&& \!\!\!\!\!\! \!\!\!\!\!\! {\rm Im}\, \ad_\Lambda = {\rm Span}_{\mathbb C[\l,\l^{-1}]} \le\{
H_i\, \Lambda^j\,| \, i=1,\dots,n,~ j=0,\dots, n
\ri\},\\
&& \!\!\!\!\!\! \!\!\!\!\!\! \Tr \le(X \Lambda^j [\rho^\vee, \Lambda_\ell]\ri) =
\le\{
\begin{array}{cr}
-h\, \lambda \, x_j \, \delta_{n+1, j+\ell} & j=1,\dots, n,\\
0 & j=0,
\end{array}
\ri.~~ \mbox{for } X = \sum_{i=1}^n x_i \, H_i.
\eea
Thus if $A_k\,(k\geq 1)$ is of form
$
A_k = \sum_{j=0}^n X^{(j)}_k(\lambda) \,  \Lambda^j(\lambda),~ X^{(j)}_k= \sum_{i=1}^n x_{k,i}^{(j)} \, H_i ~(0\leq j \leq n),$ then
\be
B_k(\lambda)= \frac {1} {  (h+1) k}  \sum_{j=1}^n x^{(j)} _{k,j}(\lambda) \, \Lambda_j(\lambda),\qquad k\geq 1.
\ee
Denote by $\zeta=\exp({2\pi \sqrt{-1}}/{h})$ a root of unity. The map $\ad_\Lambda$ acts on ${\rm Im}\,\ad_\Lambda$ as follows
\bea
&& \ad_\Lambda (X\Lambda^j) = \wt X \Lambda^{j+1},\,j=0,\dots,n,~ X,\, \wt X\in\h, \\
&&  \le(
\begin{array}{r}
\tilde{x}_1\\
\tilde{x}_2\\
\vdots\\
\tilde{x}_n\\
\end{array}\ri)
= - \le(
\begin{array}{rrrrrrrrrr}
2 & -1 & 0 & \dots&&0\\
1 & 1 & -1 & \dots&&0\\
1& 0 & 1 &- 1& \dots &0\\
\vdots &&&\ddots&\ddots\\
1 & 0 & \dots && 1& -1\\
1&0&\dots &&0&1
\end{array}
\ri) \le(
\begin{array}{r}
x_1\\
x_2\\
\vdots\\
x_n\\
\end{array}\ri).\label{81}
\eea
Define
$
\Omega_\nu  = {\rm diag} \le(\zeta^{\nu(j-1)}  \ri)_{j=1\dots n+1},
~  \nu  = 1,\dots n,
$
then we have
\be
\ad_{\Lambda} (\Omega_\nu \, \Lambda^j) = (\zeta^\nu -1) \, \Omega_\nu \, \Lambda^{j+1},\quad j=0,\dots,n.
\ee
A direct computation gives 
\bea
\&\& \rho^\vee = \sum_{\nu=1}^{n} \frac {\zeta^\nu}{ \zeta^\nu-1} \Omega_\nu,  \quad \ad_{\rho^\vee} \Lambda^0=0,\quad \ad_{\rho^\vee} \Lambda^\ell = 
\le(\sum_{\nu=1}^{n}  \frac {\zeta^\nu
- \zeta^{\nu(\ell+1)}
 }{ \zeta^\nu-1} \Omega_\nu  \ri)\Lambda^\ell,~~\ell\geq 1;
 \\
\&\& \ad_{\Lambda}^{-1} \ad_{\rho^\vee} \Lambda^{\ell}  =
 \le(\sum_{\nu=1}^{n}  \frac {\zeta^\nu
- \zeta^{\nu(\ell+1)}
 }{ \le(\zeta^\nu-1\ri)^2} \Omega_\nu  \ri)\Lambda^{\ell-1},~~\ell\geq 1.
\eea
Using these formul\ae\ and consider the leading term of $M$ given by $M_0 = \Lambda_a,$
one immediately obtains the formula of the first order approximation \eqref{first-order}.  
\end{example}

\section{Essential series and dual essential series} \label{c-e-s}
In this section we apply a reduction approach and the technique of Laplace--Borel transform 
to prove Thm.\,\ref{ce-thm}, Prop.\,\ref{S-prop} and Thm.\,\ref{ce-p}. 

First of all, we introduce the lowest weight structure constants.
\bd 
Fix a choice of $\gamma^1,\dots,\gamma^n$ and $I_+$ as above. 
The coefficients $c^{pm}_{qijs}$ uniquely determined by
\be\label{str-c}
\le[\ad_{I_+}^p\, \gamma^q, \, \ad_{I_+}^m \gamma^i \ri]=\sum_{j=1}^n \sum_{s=0}^{2m_{j}} c^{pm}_{qijs} \, \ad_{I_+}^s \gamma^j.
\ee
are called the lowest weight structure constants of $\g$.
\ed
In this paper we will only use part of these structure constants, namely $c^m_{ijs}:=c^{0m}_{nijs}$. The above definition \eqref{str-c} becomes
\be\label{str}
\le[\gamma^n, \, \ad_{I_+}^m \gamma^i \ri]=\sum_{j=1}^n \sum_{s=0}^{2m_{j}} c^m_{ij s} \, \ad_{I_+}^s \gamma^j.
\ee
\begin{lemma}\label{vanishing0}
The constants $c^m_{i j s}$ are zero unless 
\be
s-m_{j}=m-m_i-(h-1).
\ee
\end{lemma}
\begin{proof} 
By using the fact that 
$$\deg\,\ad_{I_+}^{s}\, \gamma^j=s-m_j,\qquad \forall \, s\in\{0,\dots,2m_j\}$$ 
and 
by comparing the principal degrees of both sides of \eqref{str}.
\end{proof}
\bl  \label{vanish}
If $i,j\in\{1,\dots,n\}$ satisfy $m_i+m_j=h-1$ then 
$c_{ij0}^{2m_i}=0.$
\el
\begin{proof}
It suffices to consider the case $i\leq n-1$.  By Lemma \ref{vanishing0} we have
\be
\le[\gamma^n, \, \ad_{I_+}^{2m_i} \gamma^i \ri]=\sum_{0\leq m_i+m_j-(h-1)\leq 2m_j} \, c^{2m_i}_{ij,m_i+m_j-(h-1)} \, \ad_{I_+}^{m_i+m_j-(h-1)} \gamma^j.
\ee
We will use the following orthogonality conditions \cite{Kostant} 
\be\label{ortho}
(\gamma^{i_1}\,|\,\ad_{I_+}^{2m_{i_2}} \gamma^{i_2} )=0,\qquad  \mbox{if } i_1\neq i_2.
\ee
It follows that
\be
\le(\rho^\vee \, \bigg| \le[\gamma^n, \, \ad_{I_+}^{2m_i} \gamma^i \ri]\ri)=\le([\rho^\vee,\gamma^n] \, \bigg | \, \ad_{I_+}^{2m_i} \gamma^i \ri)=-(h-1)\,\le(\gamma^n \, \bigg | \, \ad_{I_+}^{2m_i} \gamma^i \ri)=0.
\ee
On another hand, for $m_i+m_j-(h-1)\geq 2$ we have
\be
\le(\rho^\vee\,\bigg|\,\ad_{I_+}^{m_i+m_j-(h-1)} \gamma^j\ri)=\le(I_+\,\bigg|\,\ad_{I_+}^{m_i+m_j-(h-1)-1} \gamma^j\ri)=0.
\ee
For $m_i+m_j-(h-1)=1$ we have
\be
\le(\rho^\vee\,\big|\,\ad_{I_+} \gamma^j\ri)=(I_+\,|\,\gamma^j)=0.
\ee
The last equality is due to $i\neq n$ and so $j$ cannot be $1$, and $I_+=const \cdot \ad_{I_+}^2\,\gamma^1$. Hence 
$c_{ij0}^{2m_i}=0.$ The lemma is proved.
\end{proof}

\noindent {\it Proof of Theorem \ref{ce-thm}}.
Let $M$ be any solution of \eqref{int1}. It can be uniquely decomposed in the form
\be
M=\sum_{i=1}^n \sum_{m=0}^{2m_{i}} K_{im} \,\ad_{I_+}^m \gamma^i
\ee
for some functions $K_{im}=K_{im}(\lambda)$.
Substituting this expression into equation \eqref{int1} we have
\be\label{master}
\sum_{i=1}^n \sum_{m=0}^{2m_{i}} K_{im}' \, \ad_{I_+}^m \gamma^i+\sum_{i=1}^n \sum_{m=1}^{2m_{i}} K_{i,m-1} \ad_{I_+}^{m} \gamma^i+\lambda\,\le[\gamma^n, \, \sum_{i=1}^n \sum_{m=0}^{2m_{i}} K_{im} \ad_{I_+}^m \gamma^i \ri]=0.
\ee
Here $'$ means the derivative w.r.t. $\lambda.$  
Substituting \eqref{str} into \eqref{master} we find
\be
\sum_{j=1}^n \sum_{s=0}^{2m_{j}} (K_{js}' +K_{j,s-1} ) \, \ad_{I_+}^{s} \gamma^j+
\lambda \, \sum_{j=1}^n \sum_{s=0}^{2m_{j}} \sum_{i=1}^n \sum_{m=0}^{2m_{i}} K_{im} \, c_{i j s} \, \ad_{I_+}^s \gamma^j=0
\ee
where $K_{j,-1}:=0.$ Now using Lemma \eqref{vanishing0} we find
\be\label{master2}
K_{j,s-1}=-K_{js}'  -  \lambda \,  
\sum_{ m_i \geq s-m_{j}+(h-1)}  \, c^{s+m_i-m_{j}+(h-1)}_{i j s} \, K_{i,s+m_i-m_{j}+(h-1)}
\ee
where $0\leq s\leq 2m_j.$
Due to the inequality $m_i-m_j+h-1\geq m_i >  0$ the system \eqref{master2} is of triangular form. So we can solve it for  the vector $\le(K_{j,s-1}\ri)_{j=1}^{n}$ 
as long as the vectors $\le(K_{j s'}\ri)_{j=1}^n,\  s'\geq s+1$ are known.
Note that, by definition $S_{j}=K_{j,\,2m_{j}}$. Thus
$K_{im},\,m=0,\dots, 2m_{i}-1$ are determined as linear combinations of $S_{1},\dots,S_{n}$ and their derivatives. Further noticing that
\be\label{ine-s}
s+m_i-m_{j}+(h-1)\geq s+m_i,\quad \forall\,i,j
\ee
we obtain the more subtle expressions \eqref{Kim-S}.
Indeed, it is easy to see that
\be\label{simple-K}
K_{im}=(-1)^m \, S_i^{(2m_i-m)},\qquad m_i\leq m\leq 2m_i-1.
\ee
This observation together with \eqref{ine-s},\,\eqref{master2} yields \eqref{Kim-S}. Note that the coefficients $k^0_{imuv}$ and $k^1_{imuv}$ in \eqref{Kim-S} are uniquely specified by the lowest weight structure constants of $\g.$

Taking $s=0$ in Eqs. \eqref{master2} we obtain
a system of linear ODEs \eqref{S-red} of the form
\be\label{S-red-proof}
S_i^{(2m_i+1)}=\sum_{u=1}^n \sum_{ v=0}^{ 2m_i}  k_{iuv}^0  \, S_u^{(v)}+ \lambda\, \sum_{u=1}^n \sum_{ v =0 }^{ m_i}  k_{iuv}^1 \, S_u^{(v)}, \quad i=1, \dots, n
\ee
where the constant coefficients $k^0_{iuv}$ and $k^1_{iuv}$ are determined by the lowest weight structure constants of $\g$. 
It is now clear that there is a one-to-one correspondence between solutions of the topological ODE \eqref{int1} and 
solutions of the system \eqref{S-red-proof}. Thus these two systems of ODEs are equivalent.

It remains to show that $k^1_{iu0}=0,\,\forall \,i,u\,\in\{1,\dots,n\}.$
The $(s=0)$-equations in \eqref{master2} read as follows
\be \label{S-red-K}
K_{j0}'  +  \lambda \,  
\sum_{ m_i +m_{j} \geq h-1}  \, c^{m_i-m_{j}+(h-1)}_{i j 0} \, K_{i,m_i-m_{j}+(h-1)}=0.
\ee
Taking the $\lambda$-derivative in the $(m=0,\,i=j)$-equations of \eqref{Kim-S} we obtain
\bea
\&\& K_{j0}' =S_j^{(2m_j+1)}+\sum_{u=1}^n \sum_{ v=0}^{ 2m_j-1}  k_{j0uv}^0  \, S_u^{(v+1)}+\sum_{u=1}^n \sum_{ v =0}^{ m_j-1}  k_{j0uv}^1 \, S_u^{(v)}\nn\\
\&\& \qquad\qquad\qquad+ \lambda\, \sum_{u=1}^n \sum_{ v =0}^{ m_j-1}  k_{j0uv}^1 \, S_u^{(v+1)}. \label{Kj0'}
\eea
From equations \eqref{Kj0'},\,\eqref{S-red-K},\,\eqref{S-red-proof} we know that $K_{j0}'$ does not contribute to $k^1_{iu0}$.
The only possible terms that can contribute to $k^1_{iu0}$ are
\be\label{part-p10}
\sum_{ m_i +m_{j} \geq h-1}  \, c^{m_i-m_{j}+(h-1)}_{i j 0} \, K_{i,m_i-m_{j}+(h-1)}
\ee
Noticing again the inequalities
$
m_i-m_j+(h-1)\geq m_i
$
as well as \eqref{simple-K}, we find that
\eqref{part-p10} do not contribute to $k^1_{iu0}$ unless
\be
m_i-m_j+(h-1)=2m_i.
\ee
However, by Lemma \ref{vanish} we know that $c_{ij0}^{2m_i}=0$ if $m_i+m_j=h-1.$ Hence $k^1_{iu0}\equiv0.$

The theorem is proved. 
\QED

Equations \eqref{S-red} can be written in the following form (which will be
used later)
\be\label{S-red-matrix}
 \sum_{v=0}^{2h-1} p^0_{v} \, S^{(v)}+\lambda  \,  \sum_{v=1}^{h} p^1_{v} \, S^{(v)}=0,\qquad S=(S_1,\dots,S_n)^T
\ee
where $p^0_{v}$ and $p^1_{ v}$ are constant $n\times n$ matrices,  $(p^0_v)_{iu}=k^0_{iuv},\,(p^1_v)_{iu}=k^1_{iuv}.$
These matrices satisfy $(p^0_{2m_i+1})_{\,i, \,2m_i+1}=-1$ and the following vanishing conditions
\bea\label{vani}
&&
(p^0_v)_{iu}=0 \quad {\rm if}\quad v>2m_i+1,\\
\nn
&&
(p^1_v)_{iu}=0\quad {\rm if} \quad v>m_i.
\eea

Before proving Prop.\,\ref{S-prop} and Thm.\,\ref{ce-p} we do some preparations.
Recall that the essential series $S_{a;i}$ of $\g$ are defined as the $S$-coefficients of 
the generalized Airy resolvents $M_a.$ 
\begin{lemma}   \label{Sai}
For any $a\in \{1,\dots,n\}$, the essential series $S_{a;i},\,i=1,\dots,n$ have the form
\be\label{sai}
S_{a;i} = \lambda^{-\frac{s_{ai}} h} \, \sum_{m=0}^\infty \, c^a_{i,m} \, \lambda^{-m(h+1)},\quad c^{a}_{i,m}\in \CC
\ee
where $s_{ai}$ is a positive integer satisfying
\be \label{mod}
s_{ai}\equiv m_a ~~ {\rm mod} \,h,\qquad s_{ai}\equiv m_i ~~ {\rm mod} \,(h+1).
\ee
In particular, the essential series $S_{a;a}$ is non-zero and has the form
\be\label{Saa}
S_{a;a} = \lambda^{-\frac{m_a}h} \, \sum_{m=0}^\infty \, c^a_{m} \, \lambda^{-m(h+1)},\quad c^a_{0}\neq 0,~c^a_{m}\in \CC.
\ee
\end{lemma}
\begin{proof} 
Fix an integer $a\in\{1,\dots,n\}.$
By \eqref{baza1} we know that 
$S_{a;i}\in \, \lambda^{-\frac{m_a}h} \cdot \mathbb{C}((\lambda^{-1})), \quad i=1,\dots,n,$
where $\mathbb{C}((\lambda^{-1}))$ denotes the set of formal Laurent series at $\lambda=\infty.$ 
If $S_{a;i}$ is zero then it has the form \eqref{sai} with $c_{i,m}^a=0,\,m\geq 0.$
Let $S_{a;i}$ be non-zero. Recall the expansion \eqref{baza1}
\begin{eqnarray}
\&\&
M_a(\lambda) = \lambda^{-\frac{m_a}{h}} \left[ \Lambda_{m_a} +\sum_{k=1}^\infty M_{a,k}(\lambda)\right],\nn\\
\&\&
M_{a,k}(\lambda)\in L(\g),\quad \deg M_{a,k}(\lambda)=m_a-(h+1)k, \quad a=1, \dots, n.\nn
\nn
\eea
Let $S_{a;i,k}(\lambda)$ denote the $S$-coefficients of $M_{a,k}(\lambda)$, we have
\be
S_{a;i}=\lambda^{-\frac{m_a}h}\,\sum_{k=0}^\infty \, S_{a;i,k}(\lambda).
\ee
Noticing that
$
\deg \, \ad_{I_+}^{2m_i}\gamma^i=m_i
$
we obtain 
\be\deg \,S_{a;i,k}(\lambda) = m_a-m_i-(h+1)k,\quad k\geq 0.\ee
Since $\deg\,\lambda=h$ we conclude that 
\be
m_a-m_i-(h+1)k \equiv 0~({\rm mod} \, h).
\ee
Hence $S_{a;i,k}=0$ unless there exists an integer $u$ such that 
$
k=m_a-m_i+u \, h,
$
and in this case 
\be
S_{a;i,k}=const\cdot \lambda^{-(m_a-m_i)-u(h+1)},
\ee
which yields \eqref{mod}. Here, note that $k\geq 0$ implies $u\geq 0$, hence $s_{ai}>0.$

To prove \eqref{Saa} it suffices to show that $c^a_{0}\neq 0$. It has been proved by Kostant that the elements $\Lambda_{m_a},\,a=1,\dots, n$ have the 
following decompositions\footnote{We would like to thank Yassir Dinar for bringing  this useful result of \cite{Kostant} to our attention.} (cf. p.\,1014 in \cite{Kostant})
\be
\Lambda_{m_a}= \lambda \, \beta+ v\, \ad_{I^+}^{2m_a}\, \gamma^a,\quad \deg \beta=-(h-m_a)
\ee
where $v$ are some non-zero constant and $\beta\in\g$. Here, we have also applied the facts
\be 
\Ker \,\ad_{I^+}={\rm Span} \{\ad_{I^+}^{2m_1} \, \gamma^1,\dots,\ad_{I_+}^{2m_n}\,\gamma^n\},\qquad \deg\, \ad_{I_+}^{2m_i}\, \gamma^i=m_i.
\ee
The lemma is proved.
\end{proof}

Let us now prove Prop.\,\ref{S-prop}.
\begin{proof} 
The statement 1) follows from \eqref{sai}. 
The statement 2) follows from \eqref{Kim-S}   because $M_a$ are linearly independent.
The statement 3) is contained in Lemma \ref{Sai}.
Let $M$ be a regular
solution to \eqref{int1} such that the $S_a$-coefficient of $M$ coincides with $S_{a;a}$.
If the exponents of $\g$ are pairwise distinct then $M$ must be equal to $M_a$ since $S_{a;a}$ is non-zero.
Slight modifications can be done for $D_n$ with $n$ even as described in the content of the proposition.
The proposition is proved. 
\end{proof}

\noindent 
For any $a\in\{1,\dots,n\}$, Thm.\,\ref{ce-thm} reduces $M_a$ consisting of $n(h+1)$ scalar series to $n$ scalar series $S_{a;j},\,j=1,\dots,n.$
Prop.\,\ref{S-prop} further reduces $\{S_{a;j}\}_{j=1}^n$ to one scalar series $S_{a;a}$.

Finally let us consider the Laplace--Borel transform and prove Thm.\,\ref{ce-p}.
\begin{proof}
Equations \eqref{Kim-dual},\,\eqref{first-dual} follow directly from equations \eqref{Kim-S},\,\eqref{S-red}. So i) is proved.

To show ii), we will construct a basis of formal solutions to \eqref{first-dual}. Indeed, we consider the Laplace--Borel transform of the essential series $S_{a;i}(\lambda)$ 
\be\label{lb-trans-S}
S_{a;i}(\lambda)=\int_{C_a} \, \phi_{a;i} (x) \, e^{-\lambda\,x} dx
\ee
for suitable contours $C_a.$
By Lemma \ref{Sai} we have
\be
\phi_{a;i}= x^{\frac{s_{ai}}h-1} 
\, \sum_{m=0}^\infty \, 
\frac{c^a_{i,m}} {\G\le(\frac{s_{ai}}h+m(h+1)\ri)} \, x^{m(h+1)}
\ee
where the coefficients $c^a_{i,m}$ are defined by eq. \eqref{sai}.
For $i\neq a$, equations \eqref{mod} imply that the integers $s_{ai}$ satisfy
\be
s_{ai} \geq m_a+h \mbox{ and } s_{ai}\geq m_i+h+1
\ee
if all exponents are pairwise distinct; for the case $\g$ is of the $D_{n=2k}$ type, the above estimates are also true after possibly a linear change 
\bea
\&\&\phi_{n/2}\mapsto \, a_1\,\phi_{n/2}+a_2\,\phi_{n/2+1},\label{linear-a}\\
\&\&\phi_{n/2+1}\mapsto \, b_1\,\phi_{n/2}+b_2\,\phi_{n/2+1}\label{linear-b}
\eea
of two solutions $\phi_{n/2},\,\phi_{n/2+1}$.
In any case, for $i=a$ we have
\be
\phi_{a;a}= x^{\frac{m_a}h-1} 
\, \sum_{m=0}^\infty \, 
\frac{c^a_{m}} {\G\le(\frac{m_{a}}h+m(h+1)\ri)} \, x^{m(h+1)},\qquad c^{a}_0\neq 0.
\ee
Hence we have obtained a set of independent formal solutions of of \eqref{first-dual}
$$\phi_{1;j},\dots,\phi_{n;j},$$
and we collect them to a matrix solution $\Phi$ by defining $\Phi_{ia}=\phi_{a;i}$, then we have
\be \label{asym}
\Phi=\le(I+O(x)\ri)\,\le(\begin{array}{cccc}
x^{\frac{m_1}h-1} & 0 & \cdots & 0  \\
0 & x^{\frac{m_2}h-1}  & \ddots & \vdots \\
\vdots & \ddots&  \ddots & 0\\
0& 0 & \cdots & x^{\frac{m_n}h-1}\\
\end{array}\ri),\qquad x\rightarrow 0.
\ee
Note that Equations \eqref{first-dual} can be written as
\be\label{first-dual-matrix}
\sum_{v=0}^{2h-1} (-x)^v \, p_{v}^0 \, \phi+ \sum_{v=1}^{h}  (-1)^v \, p_{v}^1 \, \le(x^v\, \phi'+v\,x^{v-1} \phi \ri)=0, \quad \phi=(\phi_1,\dots,\phi_n)^T
\ee
where $p^0_v$ and $p^1_v$ are $n\times n$ matrices defined in \eqref{S-red-matrix}.
So we have
\be\label{BC}
B(x)\,\Phi'(x)=C(x)\,\Phi(x),
\ee
where $B,\,C$ are $n\times n$ matrices whose entries are polynomials in $x$ with coefficients determined uniquely by the lowest weight structure constants. 
More explicitly, the expressions of $B,C$ are given by 
$$B(x)= \sum_{v=1}^{h}  (-1)^v \, x^v\, p_{v}^1,\qquad C(x)=\sum_{v=0}^{2h-1} (-x)^v \, p_{v}^0+\sum_{v=1}^h\,(-1)^v\,v\,x^{v-1}\,p_v^1.$$
We are going to show that the matrix $B(x)$ is of the following form 
\be\label{form-B}
B=\le(\begin{array}{cccc}
0 & \cdots  & 0 & b_1 \cdot x  \\
\vdots & \reflectbox{$\ddots$} & b_2\cdot x & * \\
0 & \reflectbox{$\ddots$} &  \reflectbox{$\ddots$}  & \vdots \\
b_n\cdot x&  * & \cdots & *\\
\end{array}\ri)
\ee
where the asterisks  denote polynomials in $x$ possessing at least a double zero at $x=0.$ 
Recall that $p^1_v$ are defined as the matrix of coefficients of   $\lambda\times S_j^{(v)},\,j=1,\dots,n$ in the following equations
\be \label{Kj0prime}
K_{j0}'  +  \lambda \,  
\sum_{ m_i +m_{j} \geq h-1}  \, c^{m_i-m_{j}+(h-1)}_{i j 0} \, (-1)^{m_i+m_j+h-1} S_{i}^{(m_i+m_{j}-(h-1))}=0,\qquad '=\frac{d}{d\lambda}.
\ee
The summand with $m_i+m_j=h-1$ in \eqref{Kj0prime} does not contribute to the matrix $B(x)$, since
we have already proved in Thm.\,\ref{ce-thm} that $p^1_0=0.$ 
Let us consider the summands in \eqref{Kj0prime} with $m_i+m_j\geq h$.
Consider first the case $\g\neq D_{n}$ with $n$ even where all exponents are distinct. 
In this case the $(m_i+m_j=h)$-summand of the second term of l.h.s. of \eqref{Kj0prime} contributes to the anti-diagonal entries of $B(x)$. 
For $m_i+m_j=h+1,h+2,\dots, 2h-2$ the corresponding summands contribute to the entries of $B(x)$ marked with an asterisk. 
The first term $K_{j0}'$ also contributes to $B(x)$ through the last term of r.h.s. of \eqref{Kj0'}, namely,
$$\lambda\, \sum_{u=1}^n \sum_{ v =0}^{ m_j-1}  k_{j0uv}^1 \, S_u^{(v+1)}.$$
To understand this term we look at the $(s=1)$-equations of \eqref{master2}:
\be
K_{j0}=-K_{j1}'  -  \lambda \,  
\sum_{ m_i  + m_{j} \geq h}  \, c^{m_i-m_{j}+h}_{i j 1} (-1)^{m_i+m_j-h}\,S_i^{(m_i+m_j-h)}.
\ee
The contribution is of the form \eqref{form-B}. We proceed in a similar way with analyzing $K_{j1},K_{j2},\dots,$ in order to arrive at the matrix $B(x)$ in a finite number of steps.
In the case of $D_n$ with $n$ even, the above arguments remain valid with suitable choices of $\gamma^{n/2},\,\gamma^{n/2+1}.$

We have $\det \, B=\pm \,b_1\cdots b_n \,x^n.$ Since the dimension of the space of solutions of \eqref{int1} is finite and since 
we have already obtained $n$ linearly independent solutions of \eqref{BC}, we conclude that each $b_i$ is not zero. 
Using \eqref{vani} and noting that the $(m_i+m_j=k,\,k\geq h+1)$-off diagonal entries of $B(x)$ (which are the terms marked with asterisk) have at least a $(k-(h-1))$-tuple zero, 
we conclude that
the equations
\be
\phi'=B^{-1}\,C\,\phi,\qquad \phi=(\phi_1,\dots,\phi_n)^T
\ee
have the precise normal form \eqref{normal-form}, where the expression $V_{-1}$ is due to \eqref{asym}.

The part ii) of Theorem \ref{ce-p} is proved.

Finally, iii) follows from the normal form \eqref{normal-form} automatically. The theorem is proved.
\end{proof}

The series $\phi_{a;i}$ in the above proof are called  the {\it dual essential series} of $\g$.

The final remark is about uniqueness of essential series. We fix a principal nilpotent element $I_+$. Then, for a given choice 
of the basis $\{\gamma^i\}_{i=1}^n$ and of $\{\Lambda_{m_j}\}_{j=1}^n$, the series $\{S_{i;j}\}_{i,j=1}^n$ are determined uniquely. 
With a particular choice of only $\{\Lambda_{m_j}\}_{j=1}^n$, the element $\gamma^j$ for any $j\in\{1,\dots,n\}$ and hence  
each of the series $\{S_{i;j}\}_{i=1}^n$ is unique up to a constant factor except for the $D_n$ with even $n$ case. Because of this for presenting the expression of an essential or a dual essential series, we allow a free constant factor.
For the exceptional case $D_n$ with $n$ even, $\gamma^{n/2},\,\gamma^{n/2+1}$ and hence $\{S_{i;n/2},\,S_{i;n/2+1}\}_{i=1}^n$ are unique up to invertible linear combinations with constant coefficients. 
\section{Examples}\label{examples}
In this section we study examples of low ranks. We take a faithful matrix realization $\pi$ of $\g,$ and derive the normal form \eqref{normal-form} and 
dual essential series. Denote by $\chi$ the unique constant  such that 
\be\label{chi-dfn}
(a|b)=\chi\,\Tr\, (\pi(a)\pi(b)),\qquad \forall \,a,b\in \g.
\ee
Below, we will often write $\pi(a),\,a\in\g$ just as $a,$ for simplicity.
\subsection{Simply-laced cases}\label{exp-simply}
\begin{example} [The cases $A_1,A_2$] 
These two cases have already been studied in detail in Introduction. See \eqref{top1}--\eqref{airy3}, Example \ref{A2-intro}, Examples \ref{A1-normal},\,\ref{A2-normal}.
Here we will make a short summary.
The essential series for $A_1$ has the following explicit expression 
\be
S= \sum_{g=0}^\infty \frac{(6g-1)!!}{96^g\cdot g!}\frac1{\lambda^{\frac{6g+1}{2}}},\nn
\ee
which was derived in \cite{BDY1} by an alternative method. The dual fundamental series has the expression
\be
\phi= x^{-\frac12} \, e^{\frac{x^3}{12}}. \nn
\ee

The essential series for $A_2$ are given below in equations \eqref{513}--\eqref{516}. The dual fundamental series are given explicitly via confluent hypergeometric limit functions
\bea
\&\& \phi_{1;1}= x^{-\frac23} \cdot {}_0F_1\le(;\frac13;-\frac{x^8}{1728}\ri), \nn \\
\&\& \phi_{2;2}= x^{-\frac13} \cdot {}_0F_1\le(;\frac23; -\frac{x^8}{1728}\ri). \nn
\eea
They satisfy the following dominant ODEs, respectively:
\bea
\&\& 27 \, x^2 \, \phi_{1;1}''-81 \, x \, \phi_{1;1}'+\le(x^8-84\ri)\, \phi_{1;1}=0,\nn\\
\&\& 27 \, x^2 \, \phi_{2;2}''-27 \, x \, \phi_{2;2}'+\le(x^8-21\ri)\, \phi_{2;2}=0.\nn
\eea
\end{example}

Recall that the confluent hypergeometric limit function is defined by the convergent series
$$
_0F_1(;a;z) =\sum_{k=0}^\infty \frac1{(a)_k} \frac{z^k}{k!}
$$
where
\be\label{pochh}
(a)_k=a(a+1)\dots (a+k-1)
\ee
is the Pochhammer symbol. It is also closely related to the Bessel functions
$$
J_\nu(x) =\frac{\left( \frac{z}2\right)^\nu}{\Gamma(\nu+1)}\,  {}_0F_1\left(;\nu+1;-\frac{z^2}4\right).
$$

\begin{example} [The $A_3$ case] In this case,
$\g=sl_4(\mathbb{C}),\,h=4.$ The normal form of the dual topological ODE reads as follows
\be
\phi'=\le(\begin{array}{ccc}
-\frac34  & 0  & 0 \\
0 & -\frac12   & 0  \\
0&0&-\frac14\\
\end{array}\ri)\frac{\phi}{x}+\le(\begin{array}{ccc}
-\frac3{16} \, x^4 & 0 & -\frac{15}8\,x+\frac18\, x^6 \\
0 & -\frac14\, x^4  & 0  \\
\frac18\, x^2& 0 &  0 \\
\end{array}\ri)\,\phi.\nn
\ee

The three dominant ODEs read as follows
\bea
\&\& 64 \left(x^5-15\right) x^2\, \phi_{1;1}'' -4 \left(125-3 x^5\right) x^6 \,\phi_{1;1}' - \left(x^{15}-9 x^{10}+1134 x^5-1260\right) \phi_{1;1}=0, \nn \\
\&\& 4\, x\, \phi_{2;2}'+ (x^5+2)\, \phi_{2;2}=0, \nn \\
\&\& 64 \, x^2 \, \phi_{3;3}'' - \le(64 \,x-12 \, x^6\ri) \, \phi_{3;3}' - \le(x^{10} -18\, x^5+36\ri)\, \phi_{3;3}=0. \nn 
\eea

The dual fundamental series are given explicitly by
\bea
\&\& \phi_{1;1}= x^{-\frac34} \, e^{-\frac{3 x^5}{160}} \le[\, _0F_1\left(;\frac{1}{4};\frac{x^{10}}{4096}\right)-\frac{3x^5 }{80}\, _0F_1\left(;\frac{5}{4};\frac{x^{10}}{4096}\right) \right],\nn\\
\&\& \phi_{2;2}=x^{-\frac12} \, e^{-\frac{x^5}{20}},\nn\\
\&\& \phi_{3;3}=x^{-\frac14} \,e^{-\frac{3 x^5}{160}}  \cdot \, _0F_1\left(;\frac{3}{4};\frac{x^{10}}{4096}\right).\nn
\eea
\end{example}

\begin{example} [The $A_4$ case] \label{ex-a4} In this case $h=5$ and 
$\g=sl_5(\CC).$  The normal form of the dual topological ODE of $A_4$-type reads as follows
\be \label{norm-a4}
\phi'=\le(\begin{array}{cccc}
-\frac45  & 0  & 0 &0 \\
0 & -\frac35   & 0  &0 \\
0&0&-\frac25&0\\
0&0&0&-\frac15\\
\end{array}\ri)\frac{\phi}{x}+\le(\begin{array}{cccc}
0 & -\frac{66}{25} \, x^6 & -\frac{792}5 \, x& \frac1{10} \, x^8  \\
\frac{11}{210}\, x^4 & 0  & -\frac{12}5 \, x^6  & -\frac{11}{12} \, x \\
0& \frac{7}{300} \, x^4 &  0 & 0\\
-\frac2 {35} \, x^2 & 0 & 0 & 0\\
\end{array}\ri)\,\phi.
\ee
Let $\theta=x\,\p_x;$ then the dominant ODE for $\phi_{4;4}$ reads 
\be\label{phi4-a4}
\&\& \le[3125 \, \left(x^{12}+155\right)\,\theta^4-12500\, \left(7 x^{12}+620\right)\,  \theta^3+625 \, \left(x^{24}+1277 x^{12}+54870\right)\,\theta^2 \ri.\nn\\
\&\& \le.-1250 \, \left(x^{24}+321 x^{12}+31124\right) \, \theta + \left(x^{36}-1495 x^{24}+510995 x^{12}-9215525\right)\ri]\phi_{4;4}=0.
\ee
Noting that $x=0$ is a regular singularity of this ODE and that the indicial equation at $x=0$ reads
\be\label{roots-a4}
\le(k+\frac{1}{5}\ri)\,\le(k-\frac{11}5\ri)\le(k-\frac{23}5\ri)\le(k-\frac{47}5\ri)=0,
\ee
one then recognises that the dual fundamental series $\phi_{4;4}$ is the (unique up to a constant factor) solution corresponding to the root $k=-\frac15$,
$$
\phi_{4;4}=x^{-1/5} \left( 1-\frac{161}{2^{10}\cdot 3^5} x^{12}+\frac{26\,605\,753}{2^{23}\cdot 3^{12}\cdot 5^2} x^{24}-\dots\right).
$$
In a similar way one can treat the other three dominant ODEs.

For any solution $\phi=(\phi_1,\dots,\phi_4)^T$ of the normal form \eqref{norm-a4}, $\phi_4$ satisfies \eqref{phi4-a4}, and 
the series $\phi_{a;4},\,a=1,\dots,4$ are all non-zero.
Hence, similarly to the proof of Prop.\,\ref{S-prop}, for $A_4$, one can use
{\bf one} dominant ODE to determine all solutions of \eqref{norm-a4}.
In particular,  the interesting numbers $s_{a;4},\,a=1,\dots,4$ (see eq. \eqref{sai} above) can be read off immediately from \eqref{roots-a4}:
\be
s_{1;4}=16, ~s_{2;4}=52, ~s_{3;4}=28,~s_{4;4}=4.
\ee
We expect that this phenomenon always occurs when $h$ is a prime number.
\end{example}

\begin{example}
[The case $D_4$] \label{example-d4}
We have $n=4$, $h=6$ and $m_1=1,\,m_2=3,\,m_3=3',\,m_4=5$ and
\be
\g= \{B\in {\rm Mat} (8,\mathbb{C})\,|\, B+ S\eta B^T \eta S =0\},
\ee
where $\eta_{ij}=\delta_{i+j,9},\,1\leq i,j\leq 8$ and $S={\rm diag}\,(1,-1,1,-1,-1,1,-1,1).$ In this case, $\chi=\frac12$ and
\be
\Lambda=\left(
\begin{array}{cccccccc}
 0 & 1 & 0 & 0 & 0 & 0 & 0 & 0 \\
 0 & 0 & 1 & 0 & 0 & 0 & 0 & 0 \\
 0 & 0 & 0 & 1 & \frac{1}{2} & 0 & 0 & 0 \\
 0 & 0 & 0 & 0 & 0 & \frac{1}{2} & 0 & 0 \\
 0 & 0 & 0 & 0 & 0 & 1 & 0 & 0 \\
 0 & 0 & 0 & 0 & 0 & 0 & 1 & 0 \\
 \frac{\lambda }{2} & 0 & 0 & 0 & 0 & 0 & 0 & 1 \\
 0 & \frac{\lambda }{2} & 0 & 0 & 0 & 0 & 0 & 0 \\
\end{array}
\right).
\ee
The matrices $\Lambda_1,\,\Lambda_3,\,\Lambda_{3'},\,\Lambda_5$ can be chosen as follows
\be
\Lambda_1=\Lambda,\quad \Lambda_5=2\,\Lambda^5,\quad \Lambda_3=\Lambda^3,\quad \Lambda_{3'}=2\left(
\begin{array}{cccccccc}
 0 & 0 & 0 & 0 & 1 & 0 & 0 & 0 \\
 0 & 0 & 0 & 0 & 0 & 1 & 0 & 0 \\
 0 & 0 & 0 & 0 & 0 & 0 & 1 & 0 \\
 0 & 0 & 0 & 0 & 0 & 0 & 0 & 1 \\
 \lambda & 0 & 0 & 0 & 0 & 0 & 0 & 0 \\
 0 & \lambda & 0 & 0 & 0 & 0 & 0 & 0 \\
 0 & 0 & \l & 0 & 0 & 0 & 0 & 0 \\
 0 & 0 & 0 & \l & 0 & 0 & 0 & 0 \\
\end{array}
\right).
\ee
The normal form of the dual topological ODEs of the $D_4$-type reads as follows
\be \label{norm-d4}
\phi'=\le(\begin{array}{cccc}
-\frac56  & 0  & 0 &0 \\
0 & -\frac12   & 0  &0 \\
0&0&        -\frac12        &0\\
0&0&0&-\frac16\\
\end{array}\ri)\frac{\phi}{x}+\le(\begin{array}{cccc}
\frac{26}{27}\,x^6 & 0 & 0 & \frac{91}9\,x^3+\frac16 \, x^{10}  \\
0 & x^6  & 0   & 0 \\
0  & 0  &  x^6 & 0\\
\frac29\, x^2 & 0 & 0 & 0\\
\end{array}\ri)\,\phi.
\ee
The dual fundamental series are given by
\bea
\&\& \phi_{1;1} =x^{-\frac56}\, 
\exp\le(\frac{13 \, x^7}{189}\ri) \cdot \le[ 
\,_0F_1\left(;\frac{1}{6};\frac{x^{14}}{3^{6}}\right) 
+\frac{13} {63}\,x^{7}  \, _0F_1\left(;\frac{7}{6};\frac{x^{14}}{3^{6}}\right)\ri],\nn\\
\&\& \phi_{2;2} =\phi_{3;3}=x^{-\frac12}\, \exp \le(\frac{x^7}7\ri),\nn\\
\&\& \phi_{4;4} =  x^{-\frac16}\, \exp\le(\frac{13 \, x^7}{189}\ri) \cdot \, _0F_1\left(;\frac{5}{6};\frac{x^{14}}{3^{6}}\right).\nn
\eea
They are solutions to the following dominant ODEs, respectively:
\bea
\&\& 108 \left(3 x^7+182\right) \, x^2 \, \phi_{1;1}''-4 \left(78 \, x^{14}+5461 \, x^7+9828\right) \, x  \, \phi_{1;1}' \nn\\
\&\& \qquad\qquad\qquad-\left(12 \, x^{21}+260 \, x^{14}+107029 \, x^7+62790\right) \phi_{1;1}=0,\nn\\
\&\& \phi_{2;2}'= -\frac{\phi_{2;2}}{ 2\, x}+ x^6 \, \phi_{2;2},\qquad \phi_{3;3}'= -\frac{\phi_{3;3}}{ 2\, x}+ x^6 \, \phi_{3;3},\nn\\
\&\& 108 \, x^2 \, \phi_{4;4}''-\left(104\, x^8+108 \,x\right) \phi_{4;4}'-\left(4\, x^{14}+260 \,x^7+39\right) \phi_{4;4}=0.\nn
\eea
\end{example}

\subsection{Non-simply-laced cases}\label{exp-non-sim}

\begin{example} [The case $B_2$]
We use the realization given in \cite{DS}:
$
\g= \{B\in {\rm Mat} (5,\mathbb{C})\,|\, B+ S\eta B \eta S =0\}, \chi=\frac12
$
where $\eta_{ij}=\delta_{i+j,6}\,(1\leq i,j\leq 5),$ $S={\rm diag}\,(1,-1,1,-1,1).$ 
We have
\be
 \Lambda=\left(
\begin{array}{ccccc}
 0 & 1 & 0 & 0 & 0 \\
 0 & 0 & 1 & 0 & 0 \\
 0 & 0 & 0 & 1 & 0 \\
 \frac{\lambda}{2} & 0 & 0 & 0 & 1 \\
 0 & \frac{\lambda}{2} & 0 & 0 & 0 \\
\end{array}
\right).\nn
\ee
Recall that in this case $h=4,\,m_1=1,\,m_2=3$, and we can choose $\Lambda_{m_i}$ as follows
\be
\Lambda_1=\Lambda,\quad \Lambda_3=2\,\Lambda^3.
\ee
The normal form of the dual topological ODE is given by
\be
\phi'=\left( \begin{array}{cc} -\frac34 & 0\cr 0 & -\frac14 \end{array}\right)\frac{\phi}{x}+\left(\begin{array}{cc}\frac34 \,x^4 & \frac{15}{16} \,x+\frac{1}{4}\, x^6  \cr x^2 & 0\end{array}\right) \phi.
\ee
Solving this ODE system we obtain the dual fundamental series
\bea
\&\& \phi_{1;1}=x^{-\frac34}\, e^{\frac{3 x^5}{40}} \left[   \, _0F_1\left(;\frac{1}{4};\frac{x^{10}}{2^8}\right)+\frac{3x^5}{20}  \, _0F_1\left(;\frac{5}{4};\frac{x^{10}}{2^8}\right)\right],\nn\\
\&\& \phi_{2;2}=x^{-\frac14}\, e^{\frac{3 x^5}{40}} \, _0F_1\left(;\frac{3}{4};\frac{x^{10}}{2^8}\right).\nn
\eea
They are solutions to the following dominant ODEs
\bea
\&\& \left(64 \, x^5+240 \right) x^2\, \phi_{1;1}''-\left(48 \, x^{5}+500 \right) x^6\, \phi_{1;1}' - \left(16 \, x^{15}+36\, x^{10}+1134\, x^5+315\right) \phi_{1;1}=0,\nn\\
\&\& 16 \,x^2 \, \phi_{2;2}''-4 \left(3 \, x^5+4\right) \,x \,\phi_{2;2}'-\left(4 \, x^{10}+18 \, x^5+9\right)\, \phi_{2;2}=0.\nn
\eea
\end{example}

\begin{example} [The $B_3$ case] 
We have 
$
h=6,\,\g= \{B\in {\rm Mat} (7,\mathbb{C})\,|\, B+ S\eta B \eta S =0\},\,\chi=\frac12
$
where $\eta_{ij}=\delta_{i+j,8},\,1\leq i,j\leq 7$ and $S={\rm diag}\,(1,-1,1,-1,1,-1,1).$ The cyclic element is 
\be
\Lambda=\left(
\begin{array}{ccccccc}
 0 & 1 & 0 & 0 & 0 & 0 & 0 \\
 0 & 0 & 1 & 0 & 0 & 0 & 0 \\
 0 & 0 & 0 & 1 & 0 & 0 & 0 \\
 0 & 0 & 0 & 0 & 1 & 0 & 0 \\
 0 & 0 & 0 & 0 & 0 & 1 & 0 \\
 \frac{\lambda}{2} & 0 & 0 & 0 & 0 & 0 & 1 \\
 0 & \frac{\lambda}{2} & 0 & 0 & 0 & 0 & 0 \\
\end{array}
\right). \nn
\ee
We choose $\Lambda_{m_i}$ as follows
$
\Lambda_1=\Lambda,\quad \Lambda_5=2\,\Lambda^5,\quad \Lambda_3=\sqrt{2}\,\Lambda^3.
$
The normal form of the dual topological ODEs read
\be
\phi'=\le(\begin{array}{ccc}
-\frac56  & 0  & 0 \\
0 & -\frac12   & 0  \\
0&0&-\frac16\\
\end{array}\ri)\frac{\phi}{x}+\le(\begin{array}{ccc}
\frac{26}{27}\,x^6 & 0 & \frac{91}9\,x^3+\frac16\, x^{10} \\
0 & x^6  & 0  \\
\frac29\, x^2& 0 &  0 \\
\end{array}\ri)\,\phi,\nn
\ee
which is obviously equivalent to the normal form of $D_4$ restricted to $\phi_3\equiv0.$ Hence
\bea
\&\& \phi_{1;1} =x^{-\frac56}\, 
\exp\le(\frac{13 \, x^7}{189}\ri) \cdot \le[ 
\,_0F_1\left(;\frac{1}{6};\frac{x^{14}}{3^{6}}\right) 
+\frac{13} {63}\,x^{7}  \, _0F_1\left(;\frac{7}{6};\frac{x^{14}}{3^{6}}\right)\ri],\nn\\
\&\& \phi_{2;2} =x^{-\frac12}\, \exp \le(\frac{x^7}7\ri),\nn\\
\&\& \phi_{3;3} =  x^{-\frac16}\, \exp\le(\frac{13 \, x^7}{189}\ri) \cdot \, _0F_1\left(;\frac{5}{6};\frac{x^{14}}{3^{6}}\right).\nn
\eea
\end{example}

\begin{example}[The $G_2$ case]
We have $h=6,\,m_1=1,\,m_2=5.$ 
Using the matrix realization given in \cite{BFRFW}, we obtain 
 the normal form of the dual topological ODE of $G_2$-type
\be
\phi'=\left( \begin{array}{cc} -\frac56 & 0\cr 0 & -\frac16 \end{array}\right)\frac{\phi}{x}
+\left(\begin{array}{cc}\frac{13}{54} \,x^6 & \frac{1820}{9} \,x^3+\frac{5}{6}\, x^{10}  \cr \frac1{360} \, x^2 & 0\end{array}\right) \phi.
\ee
The dual fundamental series have the following explicit expressions
\bea
\&\& \phi_{1;1} =x^{-\frac56}\,e^{\frac{13 x^7}{756}}\,\le[\, _0F_1\left(;\frac{1}{6};\frac{x^{14}}{2^4\cdot 3^6}\ri)+
\frac{13x^7}{252}   \, _0F_1\left(;\frac{7}{6};\frac{x^{14}}{2^4\cdot 3^6}\right)\ri],\nn\\
\&\& \phi_{2;2} =  x^{-\frac16}\, e^{\frac{13 \, x^7}{756}} \cdot \, _0F_1\left(;\frac{5}{6};\frac{x^{14}}{2^4\cdot3^{6}}\right).\nn
\eea
We remark that the normal form of $G_2$-type coincides with the normal form of $B_3$ restricted to $\phi_2\equiv0$ after a rescaling of  $\Lambda$ by
$\Lambda\mapsto 2^{-\frac13}\,\Lambda.$
\end{example}

Let us summarise in Table\,\ref{nums} the numbers $m_a,r_a$ for the above examples.
\begin{table}[h]
\begin{center}
    \begin{tabular}{| c | c | c | c | }
    \hline
    $\g$ & $h$ & $m_a$                                 & $r_a$ \\ \hline
    $A_1$ & $2$ & $1$		                                & $1$\\ \hline
    $A_2$ & $3$ & $1,2$                 			       & $2,2$ \\ \hline
    $A_3$ & $4$ & $1,2,3$         			             & $2,1,2$ \\ \hline
    $A_4$ & $5$ & $1,2,3,4$     			    & $4,4,4,4$\\ \hline
    $D_4$ & $6$ & $1,3,3,5$     		          & $2,1,1,2$\\ \hline
    $B_2$ & $4$ & $1,3$              		                 & $2,2$\\  \hline
    $B_3$ & $6$ & $1,3,5$         		                      & $2,1,2$\\ \hline
    $G_2$ & $6$ & $1,5$  & $2,2$\\ \hline
    \end{tabular}
\end{center}
\caption{Simple Lie algebras and intrinsic numbers.} \label{nums}
\end{table}

\section{Concluding remarks} \label{rem}
\subsection{On the analytic properties of solutions to the topological ODEs}

We have proved the dimension of the space of solutions regular at $\lambda=\infty$ of the topological equation of $\g$-type
\be
[L,M]=0,\qquad L=\p_\lambda+\Lambda
\ee
is equal to the rank of $\g.$ We have also proved Thm.\,\ref{ce-thm} on reduction of this equation through the lowest weight gauge.

One can also consider the vector space of all solutions to the topological ODE; denote it by $\SS(\g)$.  As $\lambda=0$ is a regular point for the system \eqref{int1}, every solution $M(\lambda)\in \SS(\g)$ is uniquely determined by the initial data $M(0)$. Thus $\dim_\CC \, \SS(\g)=\dim_\CC \, \g=n(h+1).$

Denote $M^{ak}$ the unique analytic solution of \eqref{int1} determined by the following initial data
\be \label{Mak}
M^{ak}(0)=\ad_{I_+}^k \,\gamma^a,\qquad a=1,\dots,n,\,k=0,\dots,2m_a.
\ee
Clearly $M^{ak}(\lambda)$ form a basis of $\SS(\g).$ With a particular choice of basis $\{M_1,\dots,M_n\}$\footnote{Such a choice depends on a choice of an appropriate sector in the complex plane near $\lambda=\infty$.} 
of $\SS_\infty^\reg(\g)$, there exist unique {\it partial connection numbers} $C_{abk}$ such that
\be  \label{C}
M_a=\sum_{b=1}^n\,\sum_{k=0}^{2m_b}\, C_{abk}\,M^{bk}
\ee
where $C_{abk}$ are constants. 
We will study these connection numbers for the topological ODEs as well as the monodromy data for the dual topological ODEs in a subsequent publication.

\subsection{Application of 
topological ODEs
to computation of 
 intersection numbers on the moduli spaces of $r$-spin structures}

First, it should be noted that all main results of the previous sections remain valid after a change of normalization of the topological equation. It will be convenient to introduce a normalization factor $\kappa$ in the following way
\be \label{LM-Q}
[L,M]=0,\qquad L=\p_\lambda+\kappa\,\Lambda.
\ee
This parameter will be useful for applications of solutions to \eqref{LM-Q} to computation of certain topological invariants of Deligne--Mumford moduli spaces. In particular, here we explain two applications of the topological ODEs \eqref{LM-Q}. The proofs follow the scheme used in \cite{BDY1} for computing the coefficients of Taylor expansion of the logarithm  of the Witten--Kontsevich tau-function of the KdV hierarchy. Details of the proofs can be found in an upcoming publication \cite{BDY2}.

Let $\overline{\mathcal{M}}_{g,N}$ denote the Deligne--Mumford moduli space of stable curves $C$ of genus $g$ with $N$ marked points $x_1$, \dots, $x_N$, $\mathcal{L}_k$ the $k^{th}$-tautological line bundle over $\overline{\mathcal{M}}_{g,N}$, $\psi_k=c_1\left(\mathcal{L}_k \right)\in H^2\left( \overline{\mathcal{M}}_{g,N}\right)$, $k=1,\dots,N.$ 

Let $r\geq 2$ be an integer.
Fix a set of integers $\{a_k\}_{k=1}^N \subset \{0,\dots,r-1\}$  such that $2g-2-\sum_{k=1}^N a_k$ is divisible by $r$. Then the degree of the line bundle  $K_C-\sum_{i=1}^N a_i x_i$ is divisible by $r$. Here $K_C$ is the canonical class of the curve $C$. Hence there exists a line bundle ${\mathcal T}$ such that
\be
{\mathcal T}^{\otimes r}\simeq K_C-\sum_{i=1}^N a_i x_i.
\ee
Such a line bundle is not unique; there are $r^{2g}$  choices of ${\mathcal T}$.
A choice of such a line bundle ${\mathcal T}$ is called an {\it $r$-spin structure} on the curve $(C, x_1, \dots, x_N)\in\overline{\mathcal M}_{g,N}$. The space of all $r$-spin structures admits a natural compactification $\overline{\mathcal M}^{1/r}_{g; a_1, \dots, a_N}$ along with a forgetful map
$$
p:\overline{\mathcal M}^{1/r}_{g; a_1, \dots, a_N}\to \overline{\mathcal M}_{g,N}.
$$ 
There is a natural complex vector bundle ${\mathcal V}\to \overline{\mathcal M}^{1/r}_{g; a_1, \dots, a_N}$
such that the fiber over the generic point $(C, x_1, \dots, x_N, {\mathcal T})\in \overline{\mathcal M}^{1/r}_{g; a_1, \dots, a_N} $  coincides with $H^1(C, {\mathcal T})$. Here the genericity assumption means that $H^0(C,{\mathcal T})=0$. Under this assumption, using the Riemann--Roch theorem one concludes that ${\rm rank}\, {\mathcal V}=s+g-1$ where
$$
s=\frac{a_1+\dots+a_N-(2g-2)}{r}.
$$
The construction of the vector bundle ${\mathcal V}$ over non-generic points in $\overline{\mathcal M}^{1/r}_{g; a_1, \dots, a_N}$ is more subtle; for details see \cite{FSZ} and references therein.
The {\it Witten class}
$c_W(a_1,\dots,a_N)\in H^{2(s+g-1)}\left( \overline{\mathcal{M}}_{g,N}\right)$ is defined via the push-forward of the Euler class of the dual vector bundle ${\mathcal V}^\vee$
\be
c_W(a_1,\dots,a_N)=\frac1{r^g} p_*\left( {\rm euler}\left( {\mathcal V}^\vee\right)\right).
\ee
See more details on properties of $c_W(a_1,\dots,a_N)$ in \cite{W2, W3, FSZ, JKV, PV}.

Witten's $r$-spin intersection numbers are nonnegative rational numbers defined by
\bea \label{def-r-spin}
\langle\tau_{i_1,k_1}\cdots\tau_{i_N,k_N}\rangle^{r-spin}:=\int_{\overline{\mathcal{M}}_{g,N}} \psi_1^{k_1}\wedge\cdots \wedge\psi_{N}^{k_N} \wedge c_W (i_1-1,\dots,i_N-1).
\eea
Here $i_\ell \in \{1,\dots, r-1\}, \,k_\ell\geq 0,~ \ell=1,\dots, N.$ Note that we have assumed that the genus in l.h.s. of \eqref{def-r-spin} is reconstructed by the dimension counting: 
\bea \label{dim-a}
\frac{i_1-1}{r}+\cdots+\frac{i_N-1}r+\frac{r-2}{r}(g-1)+k_1+\cdots+k_N=3g-3+N. 
\eea
See  \cite{BH, BH2, LX, LVX, JKV, Zhou1,BYZ, DLM} for some explicit calculations of Witten's $r$-spin intersection numbers.

According to Witten's $r$-spin conjecture \cite{W2} the intersection numbers \eqref{def-r-spin} are coefficients of Taylor expansion of logarithm of tau-function of a particular solution to the Drinfeld--Sokolov integrable hierarchy of the $A_n$ type, $n=r-1$. This conjecture was
proved by Faber--Shadrin--Zvonkine \cite{FSZ}.

We will now use the constructed above generalized Airy resolvents for a simple algorithm for computing the $r$-spin intersection numbers in all genera.

For a given {$N\geq 1$} and a given collection of indices $i_1$, \dots, $i_N$ satisfying $1\leq i_\ell\leq r-1$ define 
the following generating function of the $N$-point Witten's $r$-spin intersection numbers by
\be\label{r-genera}F^{r-spin}_{i_1,\dots,i_N}(\lambda_1,\dots,\lambda_N)
= \left(-\kappa\,\sqrt{-r}\right)^N\sum_{k_1,\dots,k_N\geq 0}  (-1)^{k_1+\dots+k_N}  \prod_{\ell=1}^N \frac{   \left(\frac{i_\ell}{r}\right)_{k_\ell+1}}{ (\kappa\, \lambda_\ell)^{\frac{{i_\ell}}r+k_\ell+1}} \langle\tau_{i_1 , k_1} \dots \tau_{i_N , k_N}\rangle^{r-spin}.
\ee
Here $\lambda_1$, \dots, $\lambda_N$ are indeterminates, $\kappa=\le(\sqrt{-r}\ri)^{-r}$. We use the standard notation \eqref{pochh} for the Pochhammer symbol.


\bt 
\label{r-spin-thm}
Let $n=r-1,\, \g= sl_{n+1} (\CC)$,  $\Lambda=\sum_{i=1}^n E_{i,i+1}+\lambda\, E_{n+1,1}.$ 
Let $M_i=M_i(\lambda)$ be the basis of generalized Airy resolvents of $\g$-type, uniquely determined by the topological ODE 
\be
M'= \kappa \, [M,\Lambda],\qquad \kappa=\le(\sqrt{-r}\ri)^{-r}
\ee
normalized by
\be
M_i=-\lambda^{-\frac{i}h} \, \Lambda_{i}+ \mbox{lower degree terms w.r.t. } \deg,\qquad \Lambda_i:=\Lambda^i.
\ee
Denote $\eta_{ij}=\delta_{i+j,n+1}.$ {Then the generating functions \eqref{r-genera} of $N$-point correlators of $r$-spin intersection numbers have the following expressions}
\bea
&& \frac{\d F_i^{r-spin}}{\d \lambda}(\lambda)=- \kappa \, (M_i)_{1,n+1}(\lambda)- \kappa\,\lambda^{-\frac{r-1}r}\, \delta_{i,n}, \label{easy-1}\\
&& F_{i_1,\dots,i_N}^{r-spin}(\lambda_1,\dots,\lambda_N) = - \frac {1} N \sum_{s \in S_N} 
\frac {\Tr \le(M_{i_{s_1}}(\lambda_{s_1})  \dots  M_{i_{s_N}}(\lambda_{s_N}) \ri) }{\prod _{j=1}^{N} (\lambda_{s_j}- \lambda_{s_{j+1}})}\nn\\
&& \qquad \qquad\qquad\qquad\qquad\qquad\qquad -\delta_{N2} \, \eta_{i_1 i_2}  \frac {\lambda_1^{-\frac {{i_1}}h } \lambda_2^{- \frac {{i_2}}h} ({i_1} \,  \lambda_1 + {i_2} \, \lambda_2)}{(\lambda_1-\lambda_2)^2},\quad N\geq 2.
\eea
\et

{Note that up to a constant factor the entry $(M_i)_{1,n+1}$ is nothing but the essential series $S_{i;n}.$ Then, from \eqref{easy-1} we expect that all one-point $r$-spin intersection numbers can be read off from coefficients of solutions to the $n^{th}$ dominant ODE expanded near the regular singularity $x=0$.
We also remark that alternative closed expressions for one-point $r$-spin intersection numbers have been obtained in \cite{LVX} by using the Gelfand--Dickey pseudo-differential operators. }

\begin{example}
[$r=2$] Witten's $2$-spin invariants coincide with intersection numbers of $\psi$-classes over $\overline{\mathcal{M}}_{g,N}$ \cite{W1,Kontsevich,FSZ}.
So Thm.\,\ref{r-spin-thm} in the choice $r=2$
recovers the result of \cite{BDY1,Zhou2}: for $N\geq 2,$
\be\label{KdV-N-point}
&& \!\!\!\!\!\!\! {\sum_{g=0}^\infty} \sum_{p_1,\dots,p_N\geq 0} (2p_1+1)!!\cdots (2p_N+1)!!\, \int_{\overline{\mathcal{M}}_{g,N}} \psi_1^{p_1}\cdots \psi_N^{p_N}\,  \lambda_1^{-\frac{2p_1+3}{2}}\cdots\lambda_N^{-\frac{2p_N+3}2}\nn\\ 
&& \!\!\!\!\!\!\!  \qquad\qquad = - \frac{1}{N} \sum_{r\in S_{N}}  \frac{\Tr \left(M(\lambda_{r_1})\cdots 
M(\lambda_{r_N})\right)}{\prod_{j=1}^{N}(\lambda_{r_j}-\lambda_{r_{j+1}})} - \delta_{N2}\frac{\lambda_1^{-\frac12} \lambda_2^{-\frac12}(\lambda_1+\lambda_2)}{(\lambda_1-\lambda_2)^2},
\ee
where $
M=-\frac{\lambda^{-\frac12}}{2}\left(
\begin{array}{cc}
\sum_{g=1}^\infty \frac{(6g-5)!!}{(96\,\kappa^2)^{g-1}\cdot (g-1)!} \lambda^{-3g+2} & 2 \sum_{g=0}^\infty \frac{(6g-1)!!}{(96\,\kappa^2)^g\cdot g!} \lambda^{-3g}\\
\\
-2 \sum_{g=0}^\infty\frac{6g+1}{6g-1} \frac{(6g-1)!!}{(96\,\kappa^2)^g\cdot g!} \lambda^{-3g+1} &  -\sum_{g=1}^\infty \frac{(6g-5)!!}{(96\,\kappa^2)^{g-1}\cdot (g-1)!} 
\lambda^{-3g+2}\\
\end{array}
\right),~\kappa=(-2)^{-1}.
$
The formula \eqref{KdV-N-point} is closely related to a new formula recently proved by Jian Zhou \cite{Zhou3} (See Thm.\,6.1 therein). {For $N=1$, it follows easily from \eqref{easy-1} the well-known formula 
$$\langle\tau_{3g-2}\rangle_g=\frac{1}{24^g\cdot g!}\quad \mbox{for}\quad g\geq 1.$$}
\end{example}

\begin{example}
[$r=3$] Define
$
M=\left(\begin{array}{ccc} 2 K_{11} +K_{22} & - 2 S_1 -3 K_{23} & 6 S_2\\
2 K_{10}+K_{21} & -2 K_{22} & 3 K_{23}-2 S_1\\
K_{20} & 2 K_{10} -K_{21} & -2 K_{11} +K_{22}\end{array}\right),
$
where
\bea
\&\&K_{10}= \frac{S_1''}{\kappa^2}+3 \lambda  S_2, \quad K_{11}= -\frac{S_1'}{\kappa}, \nn \\
\&\&K_{20}= \frac{S_2^{(4)}}{\kappa^4}-2\lambda\,  S_1,\quad K_{21}= -\frac{S_2'''}{\kappa^3}, \quad
K_{22}= \frac{S_2''}{\kappa^2},\quad
K_{23}= -\frac{S_2'}{\kappa}. \nn
\eea 
The generalized Airy resolvents $M_1(\lambda)$ and $M_2(\lambda)$ are obtained by replacing the functions $S_1$ and $S_2$ in the above expressions by
\bea
S_{1;1}&=& \frac{1}{2^2 3^3 \kappa^2 }  \sum_{g=0}^\infty  \frac{ (-1)^{g} \cdot 3^{7g}\, \Gamma(8g+\frac13)}{24^g \cdot g! \cdot (54 \kappa^2)^{3g-1}\, \Gamma(g+\frac13)} \,  \lambda^{-\frac{24g+1}{3}},\label{513} \\
S_{1;2}&=& - \frac{1}{2^3 3^2 \kappa^2}\,\sum_{g=0}^\infty  \frac{(-1)^g  \cdot 3^{7g} \, \Gamma(8g+\frac{10}3)}{24^g\cdot g! \cdot (54\,\kappa^2)^{3g} \, \Gamma(g+\frac43)} \, \lambda^{-\frac{24g+10}{3}}
\eea
and
\bea
S_{2;1}&=& - \frac{3^4}{2^3} \,\sum_{g=0}^\infty \frac{(-1)^g\cdot 3^{7g} \, \Gamma(8g+\frac{17}3)}{24^g\cdot g! \cdot \left(54 \, \kappa^2\right)^{3 g +2} \Gamma( g+\frac53)} \lambda^{-\frac{24g+17}{3}},\\
S_{2;2}&=& - \frac{1}6\,\sum_{g=0}^\infty  \frac{ (-1)^g\cdot 3^{7g} \, \Gamma(8g+\frac23)}{24^g\cdot g! \cdot (54\,\kappa^2)^{3g}  \cdot  \Gamma(g+\frac23)} \lambda^{-\frac{24g+2}{3}},\label{516}
\eea
respectively.
Setting $\kappa=\le(\sqrt{-r}\ri)^{-r}$ we easily obtain from Theorem \ref{r-spin-thm} that 
\bea
\&\& \langle\tau_{1,8g-7}\rangle_g=\frac{1296}{46656^g \cdot (g-1)! \cdot \le(\frac13\ri)_g }\quad (g\geq 1), \nn\\
\&\&  \langle\tau_{2,8g-2}\rangle_g=\frac{1}{46656^g \cdot g!\cdot \le(\frac23\ri)_g} \quad (g\geq 1), \nn\\
\&\& F_{i_1,\dots,i_N}^{3-spin}(\lambda_1,\dots,\lambda_N) = - \frac {1} N \sum_{s \in S_N} 
\frac {\Tr\le(M_{i_{s_1}}(\lambda_{s_1})  \dots  M_{i_{s_N}}(\lambda_{s_N}) \ri) }{\prod _{j=1}^{N} (\lambda_{s_j}- \lambda_{s_{j+1}})} \nn\\
\&\& \qquad\qquad\qquad\qquad-\delta_{N2} \, \eta_{i_1 i_2}  \frac {\lambda_1^{-\frac {{i_1}}h } \lambda_2^{- \frac {{i_2}}h} ({i_1} \,  \lambda_1 + {i_2} \, \lambda_2)}{(\lambda_1-\lambda_2)^2},\quad N\geq 2. \nn
\eea
The above formulae of one-point $3$-spin intersection numbers agree with known formulae in \cite{Shadrin, BH2, LX, LVX}; e.g. the first several of these numbers are given by
$$\langle\tau_{1,1}\rangle_1=\frac1{12},\quad \langle\tau_{2,6}\rangle_1=\frac1{31104},\quad \langle\tau_{1,9}\rangle_2=\frac{1}{746496}, 
\quad  \langle\tau_{2,14}\rangle_2=\frac{1}{4837294080}.$$
\end{example}

\subsection{Topological ODEs, Drinfeld--Sokolov hierarchies, and FJRW ``quantum singularity theory"}

According to \cite{FJR,LRZ}, the partition function of FJRW invariants for an $A\,D\,E,\,D^T_n$ singularity with the maximal diagonal symmetry group is a tau-function of its 
mirror Drinfeld--Sokolov hierarchy; the partition function of FJRW invariants for a $(D_{2k},\langle J\rangle)$ singularity is a 
particular tau-function of the $D_{2k}$-Drinfeld--Sokolov hierarchy; the partition function of non-simply-laced analogue of the FJRW intersection 
numbers is defined as a tau-function of the Drinfeld--Sokolov hierarchies of $B\, C\, F\, G$-type \cite{LRZ}.

\bd
The Drinfeld--Sokolov partition function $Z$ is a tau-function of the Drinfeld--Sokolov hierarchy of $\g$-type \cite{DS} uniquely specified by the following string equation:
\be\label{string}
\sum_{i=1}^n\,\sum_{k=0}^\infty t^i_{k+1} \frac{\p Z}{\p t^i_k}+\frac12 \sum_{i,j=1}^n\, \eta_{ij} \,t^{i}_0t^j_0 Z = \frac{\p Z}{\p t^1_0}.
\ee
Here, $t^i_k$ are time variables of the Drinfeld--Sokolov hierarchy and $\eta_{ij}=\delta_{i+j,n+1}.$ 
\ed

It should be noted that the terminology ``the Drinfeld--Sokolov hierarchy of $\g$-type'' that we use refers to the Drinfeld--Sokolov hierarchy (under the choice of a principal nilpotent element) associated to the non-twisted affine Lie algebra $\hat\g^{(1)}.$ Since the construction of this integrable hierarchy does not depend on the central extension of the loop algebra, it is essentially associated with the simple Lie algebra.

Define the following generating functions of $N$-point correlators of $Z$ by
\be\label{F-DS}
F_{i_1,\dots,i_N}(\lambda_1,\dots,\lambda_N)= (-\kappa\,\sqrt{-h})^N
\sum_{k_1,\dots,k_N\geq 0}^\infty(-1)^{k_1+\dots+k_N}\prod_{\ell=1}^N 
\frac{\left( \frac{m_{i_\ell}}{h}\right)_{k_\ell+1}}{\left( \kappa\, \lambda_\ell\right)^{\frac{m_{i_\ell}}{h}+k_\ell+1}} \frac{\p^N \log Z}{\p t^{i_1}_{k_1}\dots \p t^{i_N}_{k_N}}(0).
\ee
We will now express these generating functions in terms of the generalized Airy resolvents of $\g$-type. To this end we will need to use the following multilinear forms on the Lie algebra
\be
B(a_1,\dots,a_N)=\tr \, (\ad_{a_1}\circ \dots \circ \ad_{a_N}),\qquad \forall\,a_1,\dots,a_N\in \g.
\ee

\bt \label{N-DS-p-real} Let $\g$ be a simple Lie algebra of rank $n$. Let $M_i=M_i(\lambda)$, $i=1, \dots, n$ be the generalized Airy resolvents of $\g$-type, which are the unique solutions to
\be
M'= \kappa \, [M,\Lambda],\qquad \kappa=\le(\sqrt{-h}\ri)^{-h}
\ee
subjected to 
\be
M_i=-\lambda^{-\frac{m_i}h} \, \Lambda_{m_i}+ \mbox{lower degree terms w.r.t. } \deg.
\ee
{Then the generating functions \eqref{F-DS} for the $N$-point correlators of the Drinfeld--Sokolov partition function associated to $\g$ are given by the following expressions
\bea
&& \!\!\!\!\! \frac{\d F_i}{\d \lambda}(\lambda)=- \frac{\kappa}{2\, h^\vee} B(E_{-\theta}, M_i)-\kappa\,\lambda^{-\frac{h-1}h}\, \delta_{i,n},\\
&&  \!\!\!\!\!  F_{i_1,\dots,i_N}(\lambda_1,\dots,\lambda_N) = - \frac 1 {2 N\, h^\vee} \sum_{s \in S_N} 
\frac {B\le(M_{i_{s_1}}(\lambda_{s_1}),\dots, M_{i_{s_N}}(\lambda_{s_N}) \ri) }{\prod _{j=1}^{N} (\lambda_{s_j}- \lambda_{s_{j+1}})}\nn\\
&&  \qquad\qquad\qquad\qquad\qquad\qquad -\delta_{N2} \, \eta_{i_1 i_2}  \frac {\lambda_1^{-\frac {m_{i_1}}h } \lambda_2^{- \frac {m_{i_2}}h} (m_{i_1} \,  \lambda_1 + m_{i_2} \, \lambda_2)}{(\lambda_1-\lambda_2)^2}, \quad N\geq 2. \label{24-FJR}
\eea
Here, $h^\vee$ is the dual Coxeter number.}
\et
{Similarly as before, we expect that all one-point correlators of the Drinfeld--Sokolov partition function can be read off from coefficients of solutions to the $n$-th dominant ODE expanded at $x=0$.}
\noindent The proof will be given in \cite{BDY2}. It is interesting to mention that for the $ADE$ cases, the correction term appearing in \eqref{24-FJR} for $N=2$
coincides with the propagators derived in \cite{LYZ} in the vertex algebra approach \cite{BM} to FJRW invariants; in these cases, 
the Drinfeld--Sokolov partition function coincides with the total descendant potential \cite{G1,DZ-norm,DZ1} of the corresponding Frobenius manifolds \cite{Du1,Du2}.
It would also be interesting to investigate relations between the dual fundamental series and the hypergeometric functions considered in \cite{HO}.

{In subsequent publications we will continue the study of topological ODEs as well as their applications to computation of invariants of the Deligne--Mumford moduli spaces considering generalized Drinfeld--Sokolov hierarchies \cite{GHM}.}

\begin{appendices} 

\section{Simple Lie algebras and generalized Airy functions}
Let $\g$ be an arbitrary simple Lie algebra of rank $n$. We consider in this appendix a $\g$-generalization of equation \eqref{Airy-vector}. 

Let $E_{-\theta}$ and $I_+$ be defined as in the Introduction. Let $\Lambda:=I_++\lambda\,E_{-\theta}$.
Throughout this appendix we will fix a matrix realization $\pi:\,\g\rightarrow {\rm Mat}\,(m,\CC)$. Again for any $A\in\g$ we simply write $\pi(A)$ as $A.$
\bd
The system of $m$ linear ODEs
\be \label{gen-airy}
\frac{d \vec{y}}{d\lambda}+\Lambda\,\vec{y}=0,\qquad \vec{y}=(y_1,\dots,y_m)^T\in \CC^m,\,\lambda\in \CC
\ee
is called the generalized Airy system of $\g$-type.
\ed
\noindent The dimension of the space of solutions to \eqref{gen-airy} is $m$. Let $Y:\mathbb{C}\rightarrow {\rm Mat}\,(m,\CC)$ be a fundamental solution matrix of 
\eqref{gen-airy}, i.e. 
\be
Y'+\Lambda\,Y=0,\qquad \det\, (Y)=const\neq 0.
\ee

\bp \label{A}
Let $A$ be any constant matrix in $\pi(\g)$. The matrix-valued function $M:=YAY^{-1}$ is a solution of the topological ODE, i.e. we have
\be
M'=[M,\Lambda].
\ee
\ep
\begin{proof}
By straightforward calculations.
\end{proof}

Noting that $\lambda=0$ is a regular point of the generalized Airy system \eqref{gen-airy}, 
we define a particular fundamental solution matrix $Y$ of \eqref{gen-airy} by using the following initial data
\be
Y(0)=I_m
\ee
where $I_m$ denote the $m\times m$ identity matrix. Then we have
\bp\label{A3}
Let $\{M^{ak}\}_{a=1,\dots,n,\,k=0,\dots,2m_a}$ denote the basis of $\SS(\g)$ defined by \eqref{Mak}.  Then 
\be
\pi(M^{ak})=Y\, \pi(\ad_{I_+}^k \gamma^a)\, Y^{-1}.
\ee
\ep
\begin{proof}
By using Prop.\,\ref{A} and by using the standard uniqueness theorem for linear ODEs.
\end{proof}
\bc \label{A4}
Let $\{M_a\}_{a=1,\dots,n}$ be a chosen (analytic) basis of $\SS_\infty^{\reg}(\g)$. Then we have
\be
\pi(M_a)=\sum_{b=1}^n \sum_{k=0}^{2m_b} \, C_{abk} \, Y\, \pi(\ad_{I_+}^k \gamma^a)\, Y^{-1}
\ee
where $C_{abk}$ are the corresponding partial connection numbers (c.f. \eqref{C}).
\ec

\begin{example} [$A_n,\,n\geq 1$] \label{exA5}  Take $\g=sl_{n+1}(\CC).$ We have $m=n+1.$ Any solution $\vec{y}=(y_1=y,\,y_2,\dots,\,y_m)^T$ of \eqref{gen-airy} satisfies that
\bea
\&\& (-1)^{n+1}\, y^{(n+1)}=\lambda\,y, \label{Pearcey}\\
\&\& y_{k}=(-1)^{k-1}\,y^{(k-1)},~k=2,\dots,m.
\eea
Solutions to \eqref{Pearcey} can be represented by a Pearcey-type integral
\be
y(\lambda)=\int_C \exp\le(\frac{x^{n+2}}{n+2} -\lambda\,x\ri) \, dx
\ee
where $C$ is a suitable contour on the complex $x$-plane. 
\end{example}

\begin{example}[$D_n,\,n\geq 4$]  \label{exA6} Take the matrix realization of $\g$ as in \cite{DS}. We have $m=2n.$ The cyclic element $\Lambda$ takes the form
\be 
\Lambda=\left(
\begin{array}{cccccccc}
 0 & 1 &  & &  & &  &  \\
   &  &  \ddots &  &  &&&\\
    &  &  & 1 & \frac12 &&& \\
    & & &  \ddots & 0  & \frac12 & & \\
    &&&& \ddots & 1 & &\\
    &&&&&& \ddots & \\
    \frac\lambda 2 &&&&&&& 1\\
 0 & \frac{\lambda}{2} & &&&&& 0 \\
\end{array}
\right).\nn
\ee
For any solution $\vec{y}=(y_1=y,\,y_2,\dots,\,y_m)^T$ to \eqref{gen-airy} we have
\bea
\&\& y^{(2n-1)}=\lambda\,y'+\frac12  \, y, \label{b2}, \label{dn-airy}\\
\&\& y_k=(-1)^{k-1}\,y^{(k-1)},\qquad k=2,\dots, n-1, \\
\&\& y_k=(-1)^{k-2}\,y^{(k-2)},\qquad k=n+2,\dots,2n-1,\\
\&\& 2\,y_n'=y_{n+1}'=(-1)^{n-1}\, y^{(n)},\quad y_{2n}=y^{(2n-2)}-\frac12 \,\lambda \,y,
\eea
Solutions to \eqref{dn-airy} can be represented by a Pearcey-type integral
\be
y(\lambda)=\int_C \,\exp\le(-\frac{x^{4n-2}}{2n-1} -\lambda\,x^2\ri) \, dx
\ee
where $C$ is a suitable contour in the complex $x$-plane. 
\end{example}

\begin{example}[$B_n,\,n\geq 2$] \label{exA7}  Take the matrix realization of $\g$ as in \cite{DS}. 
We have $m=2n+1,$ and $
\g= \{B\in {\rm Mat} (m,\mathbb{C})\,|\, B+ S\eta B \eta S =0\}
$
where $\eta_{ij}=\delta_{i+j,2n+2}\,(1\leq i,j\leq 2n+1),$ $S={\rm diag}\,(1,-1,1,-1,\dots,1,-1).$ 
We have
\be 
\Lambda=\left(
\begin{array}{ccccc}
 0 & 1 & 0 & \cdots & 0 \\
 \vdots & \ddots & \ddots & \ddots & \vdots \\
 0 & 0 & \ddots & \ddots & 0 \\
 \frac{\lambda}{2} & 0 & \ddots & \ddots & 1 \\
 0 & \frac{\lambda}{2} & 0 & \cdots & 0 \\
\end{array}
\right).\nn
\ee
Then any solution $\vec{y}=(y_1=y,\,y_2,\dots,\,y_m)^T$ of \eqref{gen-airy} satisfies
\bea
\&\& y^{(2n+1)}=\lambda\,y'+\frac12  \, y, \label{bn-airy}\\
\&\& y_{k}=(-1)^{k-1}\,y^{(k-1)},~k=2,\dots,m-1,\\
\&\& y_m=(-1)^{m-1} \, y^{(m-1)}-\frac\lambda 2 \, y.
\eea
Solutions to \eqref{bn-airy} can be represented by the following Pearcey-type integral  
\be
y(\lambda)=\int_C  \, \exp\le(-\frac{x^{4n+2}}{2n+1} -\lambda\,x^2\ri) \, dx
\ee
where $C$ is a suitable contour in the complex $x$-plane. 
\end{example}

\begin{example}[$E_6$]
We take the matrix realization the same as in \cite{DLZ}, where $m=27.$ The matrices $I_+$ and $E_{-\theta}$ can also be read off from \cite{DLZ}.
The corresponding generalized Airy system reduces to two linear ODEs for $y_1,\,y_6.$ Denote by $u(x),\,v(x)$ the Laplace--Borel transforms of $y_1,\,y_6$ respectively:
\be
y_1(\lambda)=\int_C\, u(x)\, e^{-\lambda \, x}\,dx,\qquad y_6(\lambda)=\int_C\, v(x)\, e^{-\lambda \, x}\,dx.
\ee
Then $u,v$ satisfy the following system of ODEs
\be
\le(\begin{array}{c}
u'\\
v'\\
\end{array}\ri)=\frac1x \,\le(\begin{array}{cc}
-\frac13 & 0\\
0&-\frac23\\
\end{array}\ri)
\le(\begin{array}{c}
u\\
v\\
\end{array}\ri)+\le(\begin{array}{cc}
5\,x^{12} & -\frac{26}3\, x^8\\
-\frac{26}9\, x^{16} &5\,x^{12}\\
\end{array}\ri)
\le(\begin{array}{c}
u\\
v\\
\end{array}\ri).
\ee
It follows that
\bea
\&\& -27 \, x^2 \, u''(x)+(270 \, x^{14} +189 \,x)\, u'+(x^{26} +675  \, x^{13} +75) \, u=0,\label{u-e6}\\
\&\& -27 \, x^2 \, v''(x)+(270 \, x^{14} +405 \, x)\, v'+(x^{26} -405 \, x^{13} +300) \, v=0.
\eea
Solving \eqref{u-e6} we have
\be
u=c_1 \, x^{-\frac13}\, e^{\frac{5 \, x^{13}}{13}} \, _0F_1\left(;\frac{2}{3};-\frac{x^{26}}{27}\right)+ c_2 \,x^{\frac{25}3} \, e^{\frac{5 \, x^{13}}{13}} \, _0F_1\left(;\frac{4}{3};-\frac{x^{26}}{27}\right)
\ee
where $c_1,c_2$ are arbitrary constants.
\end{example}


Applying the saddle point technique to the above integral representations one can derive asymptotic expansions for the generalized Airy functions. In principle, Proposition \ref{A} can be used for computing the asymptotic expansions of solutions to the topological ODE. However, the methods explained in the main part of the present paper seem to be more efficient for obtaining such expansions for regular solutions.

\begin{remark} For the original construction \cite{Kontsevich} of the Witten--Kontsevich tau-function it was introduced a matrix analogue of the Airy integral. More recently in \cite{FV} the Kontsevich's matrix Airy function was generalized for an arbitrary compact Lie group. For $\g=sl_{n+1}(\mathbb C)$ the generalized matrix Airy function appears in \cite{dijk} in representation of the tau-function of the Drinfeld--Sokolov hierarchy associated with the $r$-spin intersection numbers on $\overline{\mathcal M}_{g,N}$, $n=r-1$ (see above). It is not clear whether the generalized matrix Airy functions of \cite{FV} for other simple Lie algebras can be used for computing the FJRW intersection numbers.
\end{remark}

\end{appendices}

~~~~
~~~~

\noindent Marco Bertola

\noindent Department of Mathematics and
Statistics, Concordia University, 1455 de Maisonneuve W., Montr\'eal, Qu\'ebec,  H3G 1M8,
Canada

\noindent SISSA, via Bonomea 265, Trieste 34136, Italy

\noindent Centre de recherches math\'ematiques, Universit\'e de Montr\'eal, C.~P.~6128, succ. centre ville, Montr\'eal,
Qu\'ebec, H3C 3J7, Canada

\noindent marco.bertola@concordia.ca, mbertola@sissa.it, bertola@crm.umontreal.ca

~~~~~
~~~~~

\noindent Boris Dubrovin

\noindent SISSA, via Bonomea 265, Trieste 34136, Italy


\noindent dubrovin@sissa.it

~~~~~
~~~~~

\noindent Di Yang

\noindent SISSA, via Bonomea 265, Trieste 34136, Italy

\noindent dyang@sissa.it

\end{document}